\documentclass[superscriptaddress,aps,pra,nofootinbib,onecolumn,notitlepage,10pt]{revtex4-1}
\pdfoutput=1
\usepackage{graphicx}
\usepackage{amsthm}
\usepackage{thmtools}
\usepackage{mathtools}
\usepackage{thm-restate}
\usepackage{amsmath}
\usepackage{amssymb}
\usepackage{comment}
\usepackage{placeins}
\usepackage[caption=false]{subfig}
\usepackage[colorlinks]{hyperref}
\usepackage{url}
\usepackage{hypcap}
\usepackage{dsfont}
\usepackage{algorithm}
\usepackage{algorithmic}

\newcommand{\eq}[1]{Eq.~\hyperref[eq:#1]{(\ref*{eq:#1})}}
\renewcommand{\sec}[1]{\hyperref[sec:#1]{Section~\ref*{sec:#1}}}
\newcommand{\app}[1]{\hyperref[app:#1]{Appendix~\ref*{app:#1}}}
\newcommand{\tab}[1]{\hyperref[tab:#1]{Table~\ref*{tab:#1}}}
\newcommand{\fig}[1]{\hyperref[fig:#1]{Figure~\ref*{fig:#1}}}
\newcommand{\figa}[2]{\hyperref[fig:#1]{Figure~\ref*{fig:#1}#2}}
\newcommand{\figx}[2]{\hyperref[fig:#1]{Figure~\ref*{fig:#1}(#2)}}
\newcommand{\thm}[1]{\hyperref[thm:#1]{Theorem~\ref*{thm:#1}}}
\newcommand{\lem}[1]{\hyperref[lem:#1]{Lemma~\ref*{lem:#1}}}
\newcommand{\cor}[1]{\hyperref[cor:#1]{Corollary~\ref*{cor:#1}}}
\newcommand{\defn}[1]{\hyperref[def:#1]{Definition~\ref*{def:#1}}}
\newcommand{\alg}[1]{\hyperref[alg:#1]{Algorithm~\ref*{alg:#1}}}

\def\bra#1{\mathinner{\langle{#1}|}}
\def\ket#1{\mathinner{|{#1}\rangle}}
\newcommand{\braket}[2]{\langle #1|#2\rangle}

\newcommand{\be}{\begin{equation}}
\newcommand{\ee}{\end{equation}}
\newcommand{\ba}{\begin{eqnarray}}
\newcommand{\ea}{\end{eqnarray}}

\newcommand{\nn}{\nonumber \\}

\newcommand{\select}[1]{\textrm{select}(#1)}

\newtheorem{theorem}{Theorem}
\newtheorem{lemma}[theorem]{Lemma}

\newtheorem{corollary}[theorem]{Corollary}

\newenvironment{proofof}[1]{\begin{trivlist}\item[]{\flushleft\it
Proof of~#1.}}
{\qed\end{trivlist}}

\begin{document}

\title{Bounding the costs of quantum simulation of many-body physics in real space}

\date{\today}
\author{Ian D. Kivlichan}
\email{ian.kivlichan@gmail.com}
\affiliation{Department of Chemistry and Chemical Biology, Harvard University, Cambridge, MA 02138}
\affiliation{Department of Physics, Harvard University, Cambridge, MA 02138, USA}
\author{Nathan Wiebe}
\email{nawiebe@microsoft.com}
\affiliation{Station Q Quantum Architectures and Computation Group, Microsoft Research, Redmond, WA 98052, USA}
\author{Ryan Babbush}
\email{babbush@google.com}
\affiliation{Google Inc., Venice, CA 90291, USA}
\author{Al\'{a}n Aspuru-Guzik}
\email{aspuru@chemistry.harvard.edu}
\affiliation{Department of Chemistry and Chemical Biology, Harvard University, Cambridge, MA 02138}

\begin{abstract}
We present a quantum algorithm for simulating the dynamics of a first-quantized Hamiltonian in real space based on the truncated Taylor series algorithm.
We avoid the possibility of singularities by applying various cutoffs to the system and using a high-order finite difference approximation to the kinetic energy operator.
We find that our algorithm can simulate $\eta$ interacting particles using a number of calculations of the pairwise interactions that scales, for a fixed spatial grid spacing, as $\tilde{O}(\eta^2)$, versus the $\tilde{O}(\eta^5)$ time required by previous methods (assuming the number of orbitals is proportional to $\eta$), and scales super-polynomially better with the error tolerance than algorithms based on the Lie-Trotter-Suzuki product formula.  
Finally, we analyze discretization errors that arise from the spatial grid and show that under some circumstances these errors can remove the exponential speedups typically afforded by quantum simulation.  
\end{abstract}
\maketitle

\begin{comment}
To do:
\begin{itemize}
\item Discussion: One thing we could not analyze was what happens when we restrict to purely fermionic systems. The common case in chemistry is that we consider a system of $\eta$ electrons with the Coulomb interaction and a nuclear potential on top of that. In this case, since the electrons are fermions we will not reach the worst case which we've mentioned since it is only possible for bosonic wave functions. It is plausible, but seems difficult to prove, that the maximum possible value of the derivatives of a fermionic wave functions are much smaller than those for a totally general wave function.
\item Discussion: One really(!) nice thing from our paper is the application of arbitrary high order finite difference formulas for evaluating derivatives on a quantum computer, and in particular the use of it to reduce errors to $\log(1/\epsilon)$. To our knowledge this is the first time either has been demonstrated - it can be applied to pretty much any ODE/PDE solver (e.g. Yudong's Poisson or Dominic Berry's general one).
\item Discussion: Can the exponential error in $V$ be suppressed as the error in $T$ was? (Seems very unlikely.)
\end{itemize}
\end{comment}

\section{Introduction}
\label{sec:intro}

Simulation of quantum systems was one of the first applications of quantum computers, proposed by Manin \cite{manin1980computable} and Feynman \cite{feynman1982simulating} in the early 1980s. Using the Lie-Trotter-Suzuki product formula \cite{suzuki1991general}, Lloyd demonstrated the feasibility of this proposal in 1996 \cite{lloyd1996universal}; since then a variety of quantum algorithms for quantum simulation have been developed~\cite{aharonov2003adiabatic,childs2004quantum,berry2007efficient,childs2010on,wiebe2010higher,childs2011simulating,poulin2011quantum,wiebe2011simulating,berry2012black}, with applications ranging from quantum chemistry to  quantum field theories to spin models \cite{lidar1999calculating,aspuru2005simulated,jordan2012quantum,lasheras2014digital}. 

Until recently, all quantum algorithms for quantum simulation were based on implementing the time-evolution operator as a product of unitary operators, as in Lloyd's work using the Lie-Trotter-Suzuki product formula. A different avenue that has become popular recently is the idea of deviating from such product formulas and instead using linear combinations of unitary matrices to simulate time evolution~\cite{childs2012hamiltonian,berry2014exponential,berry2015simulating,low2016hamiltonian}. 
This strategy has led to improved algorithms that have query complexity sublogarithmic in the desired precision, which is not only super-polynomially better than any previous algorithm but also optimal.

Of the many methods proposed so far, we focus here on the BCCKS algorithm which employs a truncated Taylor series to simulate quantum dynamics~\cite{berry2015simulating}.  The algorithm has been applied to yield new algorithms for several problems, including linear systems \cite{childs2015quantum}, Gibbs sampling \cite{chowdhury2016quantum}, and simulating quantum chemistry in second quantization \cite{babbush2016exponentially} as well as in the configuration interaction representation \cite{babbush2015exponentially}. For this reason, it has become a mainstay method in quantum simulation and beyond.

The algorithms in~\cite{babbush2016exponentially,babbush2015exponentially} build on a body of work on the simulation of quantum chemistry using quantum computers: following the introduction of Aspuru-Guzik \textit{et al.}'s original algorithm for quantum simulation of chemistry in second quantization \cite{aspuru2005simulated}, Wecker \textit{et al.} \cite{wecker2014gate} determined the first estimates on the gate count required for it; these estimates were reduced through a better understanding of the errors involved and their origins, and the algorithm improved in several subsequent papers \cite{hastings2015improving,poulin2015trotter,mcclean2014exploiting,babbush2015chemical,reiher2016elucidating}. All these papers focused on second-quantized simulation: only a handful have considered the problem of simulating chemistry or physics in position space.

The reason for this is that second-quantized simulations require very few (logical) qubits.  Important molecules, such as ferredoxin or nitrogenase, responsible for energy transport in photosynthesis and nitrogen fixation, could be studied on a quantum computer using on the order of $100$ qubits using such methods.  Simulations that are of great value both scientifically and industrially could be simulated using a small quantum computer using these methods.  By contrast, even if only $32$ bits of precision are used to store each coordinate of the position of an electron in a position space simulation then methods such as \cite{wiesner1996simulations,zalka1998simulating,zalka1998efficient,boghosian1998simulating} would require $96$ qubits just to store the position of a single electron.  Since existing quantum computers typically have fewer than $20$ qubits, such simulations have garnered much less attention because they are challenging to perform in existing hardware.

However, there are several potential  advantages to position space simulations.  Most notably, these methods potentially require fewer gates than second-quantized methods.  In particular,  
Kassal \textit{et al.} found that position space simulation is both more efficient and accurate for systems with more than $\sim\!4$ atoms.  This is also important because the more gates an algorithm needs, the more physical qubits it requires to implement in hardware.  This potentially allows position space simulation to have advantages in space, time, and accuracy over second-quantized simulations once fault-tolerant quantum computing comes of age~\cite{jones2012faster}.  For these reasons, research has begun to delve more deeply into such simulations in recent years.

Recent work by Somma has investigated the simulation of harmonic and quartic oscillators using Lie-Trotter-Suzuki formulas \cite{somma2015quantum,somma2016trotter}.  This work highlights the challenges faced when trying to go beyond previous work using recent linear combination-based approaches \cite{berry2014exponential,low2016hamiltonian}, because the complexity of such methods depends on the norm of the Hamiltonian, which is unbounded.  This work highlights the fact that going beyond the Lie-Trotter-Suzuki formalism for these continuous simulations, as well as simulations of the Coulomb Hamiltonian, is not a straightforward extension of previous work.

However, a subject not considered in past literature is the errors incurred by discretizing a continuous system into a uniform mesh, a prerequisite for existing such quantum simulation algorithms. We revisit the encoding of Wiesner and Zalka \cite{wiesner1996simulations,zalka1998simulating,zalka1998efficient} as used in Lidar and Wang as well as Kassal \textit{et al.}'s works, conducting a rigorous analysis of the errors involved in discretizing the wave function of a many-particle system, and determining the number of grid points necessary for simulation to arbitrary precision. We find that these errors scale exponentially in the number of particles in the worst case, give an example of a wave function with this worst-case scaling, and finally discuss which cases we expect to be simpler. 
Further, we present an algorithm for position space quantum simulation of interacting particles using the BCCKS truncated Taylor series method \cite{berry2015simulating}. Our algorithm uses arbitrary high-order finite difference formulae~\cite{li2005general} to suppress the errors in the kinetic energy operator with only polynomial complexity, which seems challenging for approaches based on the quantum Fourier transform.

This paper is structured as follows. In \sec{results} we outline our results. We review the BCCKS Taylor series algorithm in \sec{simulation}. \sec{approximating} details the approximations to the Hamiltonian which we make, including imposing bounds on the potential energy and its derivatives, as well as the high-order finite difference approximation to the kinetic energy operator. In \sec{applying}, we discuss the problem of applying terms from the decomposition of the Hamiltonian into a linear combination of unitary operators. \sec{evolving} presents the complexity of evolving under the Hamiltonian. Finally, in \sec{errors} we discuss the various errors incurred by discretizing a continuous system (under various assumptions) and what is required to control them.

\section{Summary of results}
\label{sec:results}
Here we focus on simulating the dynamics of systems that have a fixed number of particles $\eta$ in $D$ dimensions, interacting through a spatially varying potential energy function $V(x):\mathbb{R}^{\eta D} \mapsto \mathbb{R}$.  We further assume that the simulation is performed on a bounded hypertorus: $x\in [0,L]^{\eta D}$.  In practice the assumption of periodic boundary conditions is just to simplify the construction of our approximate Hamiltonian, and non-periodic boundary conditions can be simulated by choosing the space $[0,L]^{\eta D}$ to be appropriately larger than the dynamically accessible space for the simulation.

Under the above assumptions, we can express the Hamiltonian for the continuous system as
\be
H = T + V,
\ee
where $T = -\sum_{i} \frac{\nabla_i^2}{2m_i}$ is the usual kinetic energy operator and $V=V(x)$ is some position-dependent potential energy operator, with $m_i$ the mass of the $i^\text{th}$ particle and $\nabla_i^2 = \sum_n \frac{\partial^2}{\partial {x_{i,n}}^2}$ for $x\in [0,L]^{\eta D}$. 
$i$ indexes the $\eta$ particles, and $n$ indexes the $D$ dimensions. 
We begin with the definition of the finite difference approximation \cite{li2005general} to the kinetic energy operator.  The kinetic operator in the Hamiltonian is not bounded, which means that simulation is under most circumstances impossible without further approximation.  We address this by discretizing the space and defining a discrete kinetic operator as follows.
\begin{restatable}{definition}{finitedifferenceapprox}
\label{def:dcoeffs}
Let $S_{i, n}$ be the centered finite difference approximation of order $2a$ and spacing $h$ to the kinetic energy operator for the $i^\text{th}$ particle along the $n^\text{th}$ dimension, and let $\tilde T = \sum_{i, n} \left(S_{i, n} - \frac{d_{ 2a + 1, j = 0 }}{2m_i h^2} \mathds{1}\right)$, where $d_{ 2a+1, j = 0 } = -\sum_{j=-a, j\neq 0}^a d_{ 2a+1, j }$, with
\be
\label{eq:dcoeffs}
d_{ 2a+1, j\neq0 } = \frac{2 (-1)^{a + j + 1} (a!)^2 }{(a + j)! (a - j)! j^2}.
\ee
\end{restatable}
Our kinetic energy operator differs from the usual discretized operator by a term proportional to the identity, $D \sum_i \frac{d_{ 2a + 1, j = 0 }}{2m_i h^2} \mathds{1}$.  Since the identity commutes with the remainder of the Hamiltonian, it does not lead to an observable difference in the dynamics and we thus neglect it from the simulation.  In cases where the user wishes to compute characteristic energies for the system, this term can be classically computed and added to the final result after the simulation.  

In order to make this process tractable on a quantum computer, we make a further discretization of the space into a mesh of hypercubes and assume that the value of the wave function is constant within each hypercube.  We take this mesh to be uniform for simplicity and assume that each of the $D$ spatial dimensions is discretized into $b$ points.  We further define the side length of the hypercubes to be $h:=L/b$.

\begin{restatable}{definition}{grid}
\label{def:grid}
Let $S=[0,L]^{\eta D}$ and let $\{\mathcal{D}_j:j=1,\dotsc,b^{\eta D}\}$ be a set of hypercubes that comprise a uniform mesh of $S$, let $\{y_j: j=1,\dotsc, b^{\eta d}\}$ be their centroids, and let $y:x\mapsto {\rm argmin}_{u\in \{y_j\}}\|x-u\|$ if the argmin is unique and define it to be the minimum index of the $y_j$ terms in the argmin if it is not.
We then define the discretized Hamiltonian via
\begin{enumerate}
\item $\tilde{V}:\mathbb{R}^{\eta D}\mapsto \mathbb{R}$ is defined such that $\tilde{V}(x)=V(y(x))$.
\item $\tilde{H}:=\tilde{T} + \tilde{V}$.
\end{enumerate}
\fig{centroid_grid} illustrates the hypercubes $\{\mathcal{D}_j\}$ and their centroids $\{y_j\}$ for a single particle with $b=5$ bins in $D=2$ dimensions.
\end{restatable}

\begin{figure}[t]
\begin{center}
\includegraphics[width=0.47\textwidth, trim={3cm 0.95cm 3cm 3.125cm},clip]{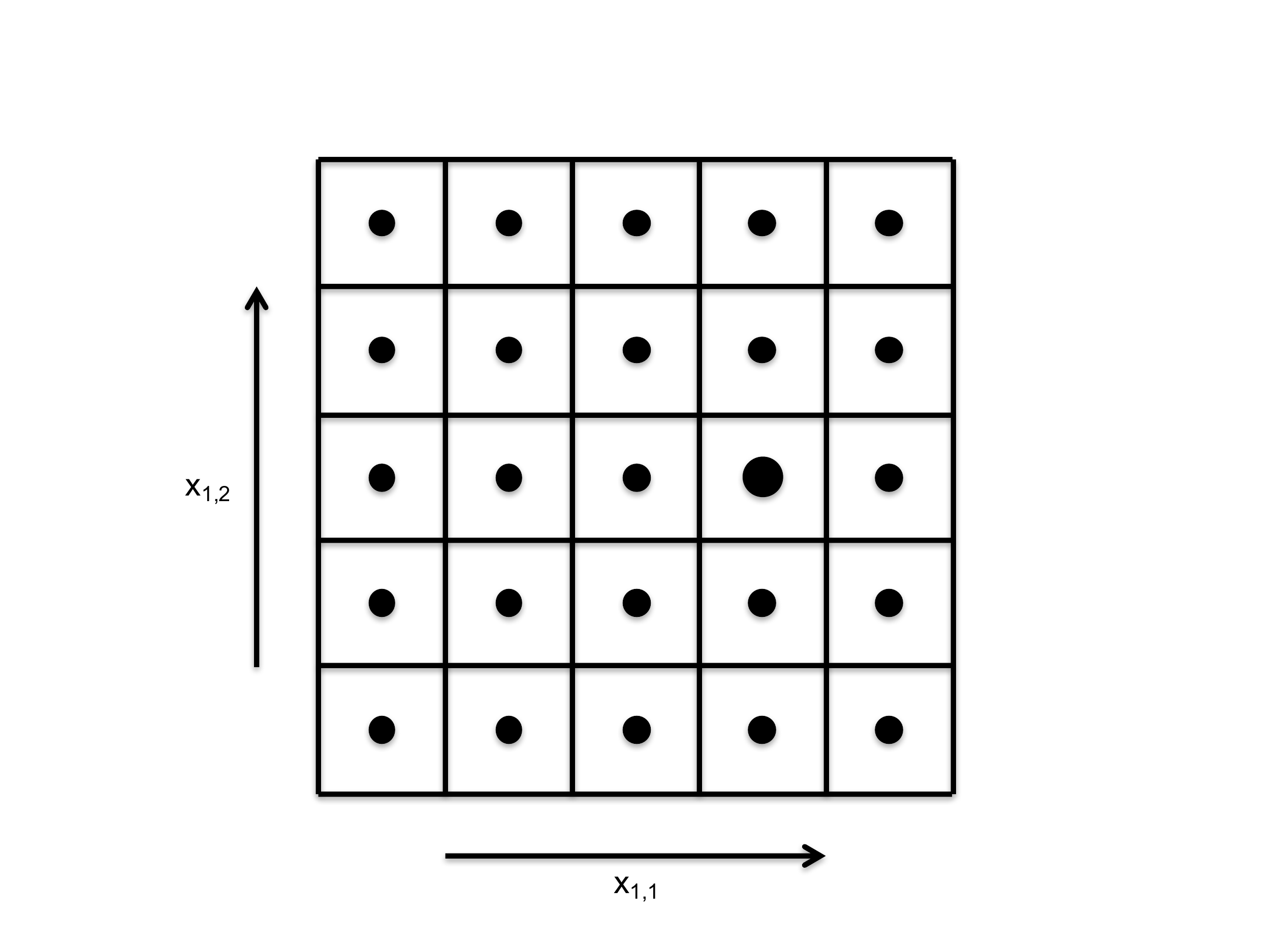}
\caption{The grid with $b=5$ bins in $D=2$ dimensions for a single particle. Each bin has side length $h = L / b$, where $L$ is the side length of the entire grid. The centroid of each bin is the dot at its centre. A particle in the bin corresponding to the larger dot would be represented by $\ket{x} = \ket{x_{1,1}} \ket{x_{1,2}} = \ket{011}\ket{010}$, indicating that $x_{1}$ is in the bin third from the left and second from the bottom, taking the bottom-left bin as $(0,0)$.}
\label{fig:centroid_grid}
\end{center}
\end{figure} 

The computational model that we use to analyze the Hamiltonian evolution is an oracle query model wherein we assume a universal quantum computer that has only two costly operations.  The first operation is the computation of the potential energy $\tilde{V}(x)$, which we cost at one query.  Furthermore, we will express our kinetic operator, approximated using finite difference formulas, as a sum of unitary adders. As such, we take the cost of applying one adder to the state to also be one query.  All other resources, including initial state preparation, are assumed to be free in this analysis.

With these definitions in hand we can state our main theorem, which provides an upper bound on the complexity of simulating such a discrete system (in a finite-dimensional Hilbert space) using the BCCKS Taylor series technique:
\begin{restatable}[Discrete simulation]{theorem}{querybound}
\label{thm:querybound}
Let $V$ be some position-dependent potential energy operator such that its max norm, $\| V(x) \|_{\infty}$, is bounded by $V_{\max}$, let $\tilde H$ be the discretized $\eta$-particle Hamiltonian in \defn{grid} with the potential energy operator $\tilde V(x) = V(y(x))$, and let $m$ be the minimum mass of any particle. We can simulate time-evolution under $\tilde H$ of the discretized wave function $\psi(y(x))$, $e^{-i\tilde H t} \psi(y(x))$, for time $t > 0$ within error $\epsilon > 0$ with 
$$O\left( \left( \frac{\eta D}{mh^2} + V_{\max}\right)t \left[\frac{\log\left( \frac{\eta Dt}{mh^2 \epsilon} + \frac{V_{\max}t}{\epsilon}\right) }{\log\left(\log\left( \frac{\eta Dt}{mh^2 \epsilon} + \frac{V_{\max}t}{\epsilon}\right)\right)} \right]\right)$$
 unitary adders and queries to an oracle for the potential energy.
\end{restatable}

We will discuss a version of the Coulomb potential modified such that it is bounded, $V_\text{Coulomb} = \sum_{i<j} \frac{q_i q_j}{\sqrt{\| x_i - x_j \|^2 + \Delta^2}}$, where $\Delta$ determines the maximum of the potential. For this potential, we can simulate discretized evolution under $\tilde H$ within error $\epsilon$ with
$$O\left( \left( \frac{\eta D}{mh^2} + \frac{\eta^2 q^2}{\Delta} \right)t\left[ \frac{\log\left( \frac{\eta Dt}{mh^2 \epsilon} + \frac{\eta^2 q^2 t}{\Delta\epsilon}\right) }{\log\left(\log\left( \frac{\eta Dt}{mh^2 \epsilon} + \frac{\eta^2 q^2 t}{\Delta\epsilon}\right)\right)}\right] \right)$$
unitary adders and queries to an oracle for potential energies, where $q$ is the maximum absolute charge of any particle.

This shows that if $h\in \omega(\eta^{-1})$ % CHECK
 then such a simulation can be performed in time that scales better with $\eta$ than the best known quantum simulation schemes in chemistry applications, for fixed filling fraction.  However, this does not directly address the question of how small $h$ will have to be to provide good accuracy.  The answer that we find is, in worst-case scenarios, that the value of $h$ needed can be exponentially small in the number of particles and can scale linearly with the error tolerance.  This is summarized in the following theorem.

\begin{restatable}[Discretizing continuous simulation]{theorem}{totalab}
\label{thm:totalab}
Let $V$ and $\tilde H$ be as in \thm{querybound} with the following additional assumptions:
\begin{enumerate}
\item the max norms of the derivatives of $V$, $\| \nabla V(x) \|_{\infty}$, are bounded by $V_{\max}^\prime$,
\item let $\psi(k):\mathbb{R}^{\eta D} \mapsto \mathbb{C}$ and $\psi(x):\mathbb{R}^{\eta D} \mapsto \mathbb{C}$ be conjugate momentum and position representations of the same $\eta$-particle wave function such that $e^{-iHs}\psi(k)$ and $e^{-i\tilde Hs}\psi(k)$ are zero if $\| k \|_{\infty} >k_{\rm max}$ for all $s\in [0,t]$,
\item $\psi(x)$ is smooth at all times during the evolution,

\item $k_{\max} L > \pi(2 e^{-1/3})^{2/\eta D}$.
\end{enumerate} Then for any square integrable wave function $\phi:S\mapsto \mathbb{C}$, we can simulate evolution for time $t > 0$ with the simulation error $\left| \int_S \phi^*(x) e^{-i H t} \psi(x) \,\mathrm{d}^{\eta D}x -\int_S \phi^*(x) e^{-i \tilde{H} t} \frac{\psi(y(x))}{\int_S | \psi(y(x)) |^2 \,\mathrm{d}^{\eta D}x } \,\mathrm{d}^{\eta D}x \right| \le \epsilon$ by choosing
$$
h \le \frac{2\epsilon}{3\eta D\left(k_{\max} +V'_{\max}t \right)}\left(\frac{k_{\max} L }{\pi} \right)^{-\eta D/2},
$$
and using the $(2a+1)$--order divided difference formula in~\defn{dcoeffs} where $$a\in O\left(\eta D \log(k_{\max} L) + \log\left(\frac{\eta^2 D^2 t k_{\max} (k_{\max} +V'_{\max}t)}{m\epsilon^2} \right) \right).$$
\end{restatable}

The modified Coulomb potential satisfies $V_{\max}^\prime \le \frac{\eta^2 q^2 \sqrt{3}}{9\Delta^2}$. With this potential,
\begin{equation}
h \le \frac{2\epsilon}{3\eta D\left(k_{\max} +\frac{\eta^2 q^2t \sqrt{3}}{9\Delta^2} \right)}\left(\frac{k_{\max} L }{\pi} \right)^{-\eta D/2}.
\end{equation}

Since the simulation scales as $O(h^{-2})$, the fact that $h \in O^\star((k_{\max} L)^{-\eta D/2})$ suggests that, without further assumptions on the initial state and the Hamiltonian, the complexity of the simulation given by~\thm{totalab} may be exponential in $\eta D$.  We further show in~\sec{errors} that this scaling is tight in that there are valid quantum states that saturate it.  This indicates that there are important caveats that need to be considered before one is able to conclude, for example, that a position space simulation will be faster than a second-quantized simulation.  However, it is important to note that such problems also implicitly exist with second-quantized simulations, or in other schemes such as configuration interaction, but are typically dealt with using intelligent choices of basis functions. Our work suggests that such optimizations may be necessary to make quantum simulations of certain continuous-variable systems practical.

One slightly stronger assumption to consider is a stricter bound on the derivatives of the wave function. In \thm{totalab}, we assumed only a maximum momentum. \cor{optimistic} determines the value of $h$ necessary when we assume that $| \psi^{(r)}(x) | \in O( k_{\max}^r / (\sqrt{2r+1}L^{\eta D/2})) )$. This assumption means that the wave function can never become strongly localized: it must at all times take a constant value over a large fraction of $S$. While this assumed scaling of the derivative of the wave function of $k_{\max}^r$ may seem pessimistic at first glance, it is in fact saturated for plane waves.  Furthermore, physical arguments based on the exponential scaling of the density of states given from Kato's theorem suggest that such scaling may also occur in low energy states of dilute electron gases.  Regardless, we expect such scalings to be common and show below that this does not lead to the exponential scaling predicted by~\thm{totalab}.

\begin{restatable}[Discretization with bounded derivatives]{corollary}{optimistic}
\label{cor:optimistic}
Assume in addition to Assumptions $1$--$3$ of \thm{totalab} that, at all times, $| \psi^{(r)}(x) | \le \beta k_{\max}^r / (\sqrt{2r+1} L^{\eta D/2}))$ for any non-negative integer $r$ where $\beta \in \Theta(1)$ and $hk_{\max} < e^{1/3}$. Then for any square integrable wave function $\phi:S\mapsto \mathbb{C}$, we can simulate evolution for time $t > 0$ with the simulation error at most $\epsilon$ by choosing
$$h \in O\left( \frac{\epsilon}{\eta D\left(k_{\max} +V'_{\max}t \right)}\right),$$
and 
$$a\in O\left(\log\left(\frac{\eta^2 D^2 t k_{\max} (k_{\max} +V'_{\max}t)}{m\epsilon^2} \right) \right).$$

\end{restatable}
Thus even if the derivatives of the wave function are guaranteed to be modest then our bounds show that the cost of performing the simulation such that the error, as defined via the inner product in~\thm{totalab}, can be made arbitrarily small using a polynomial number of queries to the potential operator using low-order difference formulas.  If we apply this method to simulate chemistry then the number of calls to an oracle that computes the pairwise potential (assuming $D$, $k_{\max}$ and $\Delta$ are fixed) scales as $\tilde{O}(\eta^7 t^3 / \epsilon^2)$.  This scaling is worse than the $\tilde{O}(\eta^5 t\log(1/\epsilon))$ scaling that has been demonstrated for methods using a basis in first (assuming the number of orbitals is proportional to $\eta$) or second quantization \cite{babbush2016exponentially,babbush2015exponentially}, which may cause one to question whether these methods actually have advantages for chemistry over existing methods.

When drawing conclusions in comparing these results, it is important to consider what error metric is being used.  To this end,~\thm{querybound} and~\thm{totalab} and~\cor{optimistic} use very different measures of the error (as seen in~\fig{error_map}).  The first strictly examines the error between the simulated system within the basis and the exact evolution that we would see within that basis.  
The latter two interpret the state within the quantum computer as a coarse-grained state in the infinite-dimensional simulation, and measure the error to be the maximum difference in any inner product that could be measured in the higher-dimensional space. This means that the $\tilde{O}(\eta^7 t^3/\epsilon^2)$ scaling should not be directly compared to the $\tilde{O}(\eta^5 t\log(1/\epsilon))$ scaling seen in existing algorithms because the latter does not explicitly consider the error incurred by representing the problem in a discrete basis.

\begin{figure}[t!]
\includegraphics[width=0.4\linewidth]{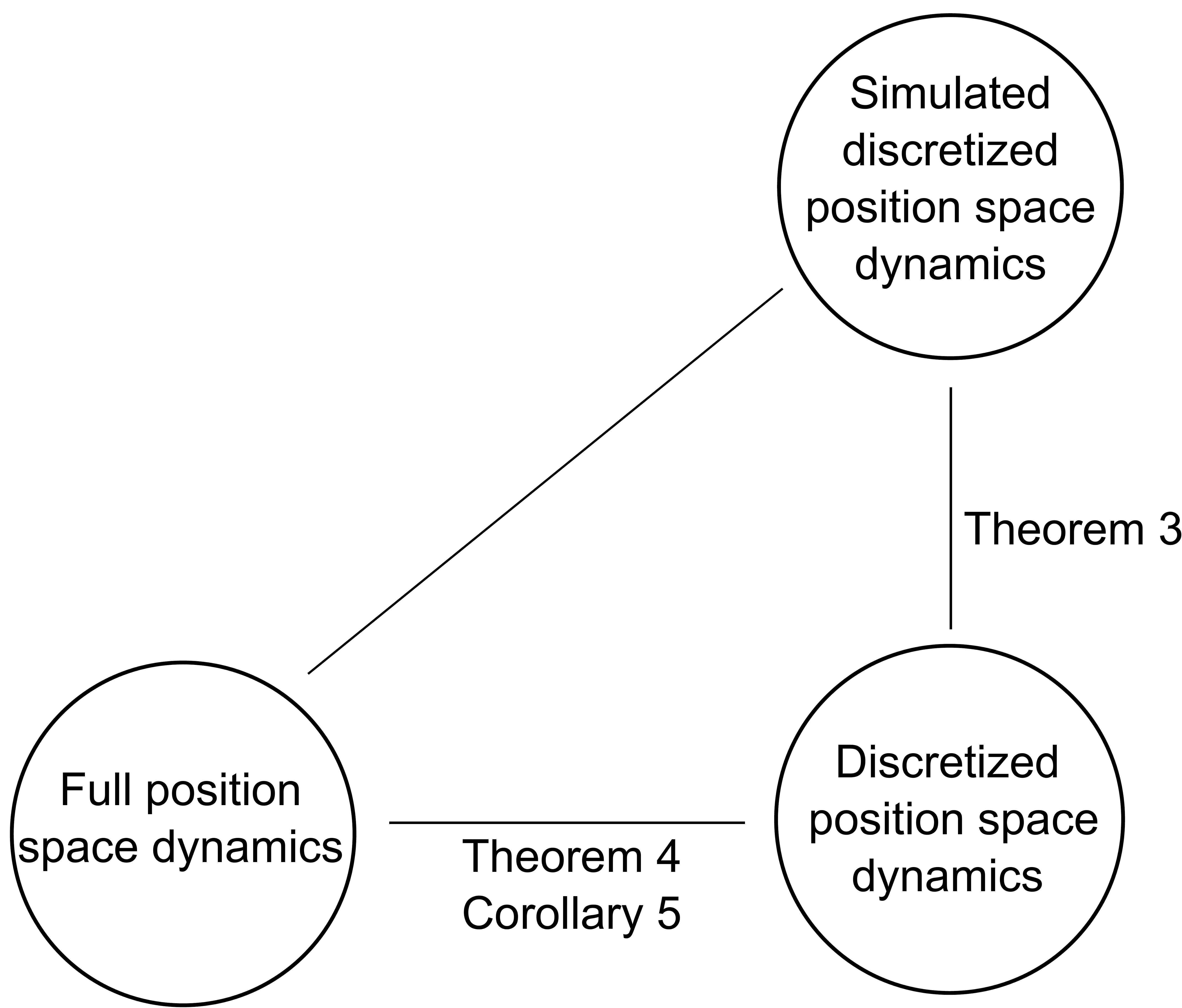}
\caption{An illustration showing the three different dynamical systems considered in this paper. The lines represent the errors incurred by the two approximations necessary for discretely simulating continuous dynamics, and are labeled by the theorems that bound them.  The overall error simulation error can be found through the use of the triangle inequality as illustrated in the figure.  Most previous results only discuss errors between the simulation and the discretized dynamics, which we bound in~\thm{querybound}.\label{fig:error_map}}
\end{figure}
 
Also, it is important to stress that~\thm{totalab} and~\cor{optimistic} bound a different sort of error than that usually considered in the quantum simulation literature.  In our setting, we assume a fixed spatial grid and allow the user to prepare an arbitrary initial state (modulo the promises made above about that state) and then discuss how badly the error can scale.  Most simulations deviate from this approach because the user typically picks a basis that is tailored to not only the state but also the observable that they want to measure.  Typical measurements include,
for example, estimation of bond lengths or eigenvalues of the Hamiltonian.  Unlike the wavefunction overlaps considered in our theorems, such quantities are not necessarily very sensitive to the number of hypercubes in the mesh.  This means that while these scalings are daunting, they need not imply that such simulations will not be practical.  Rather, they suggest that the costs of such simulations will depend strongly on the nature of the information that one wishes to extract and on the promises made about the behavior of the system.

\section{Hamiltonian simulation}
\label{sec:simulation}

To exploit the Taylor series simulation techniques, we must be able to approximate the Hamiltonian $H$ by a linear combination of easily-applied unitary operators, that is, as 
\be
\label{eq:decomp}
H \approx \sum_\chi d_\chi V_\chi,
\ee
where each $V_\chi$ is unitary, and $d_\chi > 0$. Later, we will bound the error in this approximation. As in \defn{grid}, we will work with the Hamiltonian represented in position space
\be
H = T + V,
\ee
where $T = -\sum_{i} \frac{\nabla_i^2}{2m_i}$ is the kinetic energy operator and $V=V(x)$ is the potential energy operator, with $\nabla_i^2 = \sum_n \frac{\partial^2}{\partial {x_{i,n}}^2}$. Our goal, then, is to decompose this Hamiltonian as in \eq{decomp}, into a linear combination of easily-applied unitary operators that approximates the original Hamiltonian. One way of doing this is by decomposing it into a linear combination of 1-sparse unitary operators (unitary operators with only a single nonzero entry in each row or column); however, the unitary operators need not be 1-sparse in general.

Because the potential energy operator $V=V(x)$ is diagonal in the position basis, we can decompose it into a sum of diagonal unitary operators which can be efficiently applied to arbitrary precision. The single-particle Laplacians $\nabla_i^2$ in the kinetic energy operator are more difficult: we will consider a decomposition of the kinetic energy operator, approximated using finite difference formulas, into a linear combination of unitary adder circuits. By decomposing the kinetic and potential energy operators into a linear combination of unitary operators to sufficient precision, we can decompose the Hamiltonian into unitary operators to any desired precision. The following section details how we decompose the potential and kinetic energy operators into a linear combination of unitary operators.

Once we have decomposed the Hamiltonian into a linear combination of unitary operators which can be easily applied, as in \eq{decomp}, we employ the BCCKS truncated Taylor series method for simulating Hamiltonian dynamics \cite{berry2015simulating}. We wish to simulate evolution under the Hamiltonian $H$ for time $t > 0$, that is, to approximately apply the operator
\be
U(t) = \exp(-i H t)
\ee
with error less than $\epsilon > 0$. We divide the evolution time $t$ into $r$ segments of length $t / r$, and so require error less than $\epsilon / r$ for each segment.

The key result of Ref.~\cite{berry2015simulating} is that time evolution in each segment can be approximated to within error $\epsilon / r$ by a truncated Taylor series, as
\be
\label{eq:tts}
U(t / r) = \exp(-i H t / r) \approx \sum_{k=0}^K \frac{1}{k!} (-iHt/r)^k,
\ee
where, provided that we choose $ \|H\| t$, we can take~\cite{berry2015simulating}
\be
K \in O\left(\frac{\log(r / \epsilon)}{\log\log(r / \epsilon)}\right).
\ee
Expanding \eq{tts} using the form of the Hamiltonian in \eq{decomp}, we find that
\be
\label{eq:tts_chiform}
U(t / r) \approx \sum_{k=0}^K \sum_{\chi_1, \dotsc, \chi_k} \frac{(-it / r)^k}{k!} d_{\chi_1} \dotsm d_{\chi_k} V_{\chi_1} \dotsm V_{\chi_k}.
\ee
The sum for each $\chi_i$ is over all the terms in the decomposition of the Hamiltonian in \eq{decomp}. We collect the real coefficients in the sum into one variable, $c_\alpha = (t / r)^k  / k! \prod_{k^\prime=1}^k d_{\chi_{k^\prime}}$, and the products of the unitary operators into a unitary operator $W_\alpha = (-i)^k \prod_{k^\prime=1}^k V_{\chi_{k^\prime}}$, where the multi-index $\alpha$ is
\be
\label{eq:alpha}
\alpha = (k, \chi_1, \chi_1, \dotsc, \chi_k).
\ee
We can then rewrite our approximation for $U(t / r)$ as
\be
U(t / r) \approx \sum_\alpha c_\alpha W_\alpha = W(t / r).
\ee
Since each $V_\chi$ can be easily applied, so too can each $W_\alpha$, which are products of at most $K$ operators $V_\chi$. In \sec{evolving}, we will give the circuit for an operator $\textrm{select}(W)$ such that for any state $\ket\psi$ and for any ancilla state $\ket{\alpha}$,
\be
\label{eq:selectW}
\textrm{select}(W) \ket{\alpha} \ket\psi = \ket{\alpha} W_\alpha \ket\psi.
\ee

Given a circuit for applying the operators $W_\alpha$, we can apply the approximate unitary for a single segment $W(t / r)$ using oblivious amplitude amplification \cite{berry2015simulating}.

First, we use a unitary $B$ which we define by its action on the ancilla zero state:
\be
\label{eq:B}
B \ket0 = \frac{1}{\sqrt{c}} \sum_\alpha \sqrt{c_\alpha} \ket{\alpha},
\ee
where $c = \sum_\alpha c_\alpha$ is the normalization constant $\sqrt{c}$ squared. Following this, we apply $\textrm{select}(W)$, and finally apply $B^\dagger$. Let us group these three operators into a new operator $A$, whose action on $\ket{0} \ket\psi$ is
\be
A \ket{0} \ket\psi = \frac{1}{c} \ket{0} W(t/r) \ket\psi + \sqrt{1 - \frac{1}{c^2}} \ket\phi,
\ee
where $\ket\phi$ is some state with the ancilla orthogonal to $\ket{0}$.

The desired state $W(t/r) \ket\psi$ can be ``extracted'' from this superposition using oblivious amplitude amplification~\cite{berry2014exponential,berry2015simulating}. When we allow for the fact that $W(t/r)$ may be slightly nonunitary, there are two conditions which must be satisfied in order to bound the error in oblivious amplitude amplification to $O(\epsilon / r)$ \cite{berry2014exponential}: first, we must have that $|c - 2| \in O(\epsilon / r)$, and second, we must have that \be
\label{eq:UWdiff}
\| U(t/r) - W(t/r) \| \in O(\epsilon / r),
\ee
where here and in the remainder of the paper we take $\|\cdot \|$ to be the induced $2$--norm or spectral norm.
The first condition can be satisfied by an appropriate choice of $r$, and the second is satisfied by our earlier choice $K \in O\left(\frac{\log(r / \epsilon)}{\log\log(r / \epsilon)}\right)$.

In oblivious amplitude amplification, by alternating the application of $A$ and $A^\dagger$ with the operator $R = \mathds1 - 2 P_0$ which reflects across the ancilla zero state (where $P_0$ is the projection operator onto the zero ancilla state), we construct the operator
\be
G = -A R A^\dagger R A,
\ee
which, given that $|c - 2| \in O(\epsilon / r)$ and $\| U(t/r) - W(t/r) \| \in O(\epsilon / r)$, satisfies
\be
\| P_0 G \ket{0} \ket\psi - \ket{0} U(t/r) \ket\psi \| \in O(\epsilon / r).
\ee

Thus, we can approximate evolution under the Hamiltonian $H$ for time $t / r$ with accuracy $O(\epsilon / r)$ by initializing the ancilla for each segment in the zero state, applying $P_0 G$, and discarding the ancilla. By repeating this process $r$ times, we can approximate evolution under the Hamiltonian for time $t$ with accuracy $O(\epsilon)$.

\section{Approximating the Hamiltonian}
\label{sec:approximating}

In this section, we present the approximation of the continuous Hamiltonian $H$ which we will decompose into a sum of unitary operators. We apply one approximation to the potential energy operator and two to the kinetic energy operator. To the potential energy operator $V$, we impose a cutoff on the potential energy between two particles. For the kinetic energy operator we assume a maximum momentum $k_{\max}$, and also approximate the kinetic energy operator $T$ by a sum of high-order finite difference formulas for each particle and dimension. These approximations hold for both finite- and infinite-dimensional Hilbert spaces. We focus only on the discretized finite-dimensional case because we must ultimately discretize to determine the cost of a circuit that approximates evolution under the discretized Hamiltonian $\tilde H$ in \sec{applying}.

Throughout, we employ a discrete position-basis encoding of the $\eta$-particle wave function $\psi(y(x))$. The position of each particle is encoded in $D$ registers specifying the $D$ components of that particle's position in a uniformly spaced grid of side length $L$. Each spatial direction is discretized into $b$ bins of side length $h = L / b$. We represent the stored position of particle $i$ in the $n^\text{th}$ dimension by $\ket{x_{i,n}}$, and use $\ket{x}$ to represent the combined register storing the positions of all $\eta$ particles. Each of the coordinate registers $\ket{x_{i,n}}$ is composed of $\lceil \log b \rceil$ qubits indexing which of the $b$ bins the particle is in. As such $\ket{x}$ is composed of $\eta D \lceil \log b \rceil$ qubits.

\subsection{The potential energy operator}

We first discuss the approximation to the potential energy operator $V = V(x)$. This approximation affects $V(x)$ directly, and its discretized counterpart $\tilde V(x) = V(y(x))$ through \defn{grid}. We wish to decompose the potential energy operator into a sum of unitary operators approximating $V$. Because the potential energy operator is diagonal in the position basis, this decomposition is relatively straightforward. One simple way of approximating it as a sum of unitary operators is by writing it as a sum of signature matrices, that is, diagonal matrices whose elements are $+1$ or $-1$. This requires a number of signature matrices equal to the maximum possible norm of the potential energy operator.

At this stage, the potential energy operator is unbounded, so we would need infinitely many signature matrices in the sum. 
To prevent infinities, we replace potentials of the form $\| x_i - x_j \|^{-k}$ with $1/(\sqrt{\| x_i - x_j \|^{2}+\Delta^2})^k$, with $\Delta > 0$. For example, rather than the usual Coulomb potential $\sum_{i<j} \frac{q_i q_j}{\| x_i - x_j \|}$, where $x_i$ and $q_i$ are the $D$-dimensional position and charge of the $i^\text{th}$ particle, respectively,
we instead use
\be
\label{eq:simppot}
V_\text{Coulomb} = \sum_{i<j} \frac{q_i q_j}{\sqrt{\| x_i - x_j \|^2 + \Delta^2}}.
\ee
Let $q = \max_i |q_i|$. The modified Coulomb potential energy operator is bounded by
\be
\label{eq:simppotmax}
\| V_\text{Coulomb}(x) \|_{\infty} \le \frac{\eta (\eta - 1) q^2}{2 \Delta}.
\ee
In general, we will denote the maximum value of the potential by $V_{\max}$. This means that we can approximate any bounded potential energy operator by a sum of $V_{\max}$ signature matrices. The modified Coulomb potential energy operator, for example, can be approximated by a sum of $\left\lceil \frac{\eta (\eta - 1) q^2}{2 \Delta} \right\rceil$ signature matrices. However, the error in this approximation is constant (specifically, it is at most $1$) and cannot be better controlled. We will address this issue when we discuss simulation in \sec{applying}.

\subsection{The kinetic energy operator}

In the previous subsection, we considered the problem of applying a cutoff to the potential energy operator to ensure that its norm is bounded. The kinetic energy operator has a similar problem in that its norm is not finite. Additionally, while the potential energy operator is diagonal in the position basis, the kinetic energy operator is not. This further complicates the problem of decomposing the kinetic energy operator into a linear combination of easily-applied unitary operators.

We address these issues with two simplifications. First, we approximate the kinetic energy operator using arbitrary high-order central difference formulas for the second derivative \cite{li2005general}. Second, we work only with wave functions with a maximum momentum $k_{\max}$ such that $\psi(k) = 0$ if $\| k \|_{\infty} \ge k_{\max}$. These two simplifications are linked, and, after determining a bound on the sum of the norms of the finite difference coefficients in \lem{dcoeffbound}, we will use that bound together with the momentum cutoff to bound the error incurred by the finite difference approximation in \thm{errbd}. 

We numerically approximate the Laplacian using a $(2a+1)$-point central difference formula for the second derivative of the $i^\text{th}$ particle's position in each dimension $n$. The $(2a+1)$-point central difference formula for a single such coordinate is \cite{li2005general}
\begin{equation}\begin{aligned}
\label{eq:fdf}
\partial^2_{in} \psi(x)  =& h^{-2} \sum_{j=-a}^a d_{ 2a+1, j } \psi(x + j h \hat e_{i,n} ) + O_{2a+1},
\end{aligned}
\end{equation}
where $\hat e_{i,n}$ is the unit vector along the $(i, n)$ component of $x$, $(x_{i,n} + j h \hat e_{i,n} )$ is evaluated modulo the grid length $L$, and 
\be
\label{eq:dcoeffs}
d_{ 2a+1, j\neq0 } = \frac{2 (-1)^{a + j + 1} (a!)^2 }{(a + j)! (a - j)! j^2}.
\ee
The $j=0$ coefficient is the opposite of the sum of the others, $d_{ 2a+1, j = 0 } = -\sum_{j=-a, j\neq 0}^a d_{ 2a+1, j }$.

Surprisingly, the sum of the norms of the finite difference coefficients $d_{ 2a+1, j\neq0 }$ is bounded by a constant. We prove this fact below, and then use it to bound the error term $O_{2a+1}$.
\begin{lemma}
\label{lem:dcoeffbound}
The sum of the norms of the coefficients $d_{ 2a+1, j\neq0 }$ in the $(2a+1)$-point central finite difference formula is bounded above by $\frac{2 }{3}\pi^2$ for $a\in \mathbb{Z}_{+}$.
\end{lemma}
\begin{proof}
The sum of the norms of the coefficients is
$$
\begin{aligned}
\sum_{j=-a, j \neq 0}^{a} \left| d_{ 2a + 1, j } \right| &=  \sum_{j=-a, j \neq 0}^{a}  \frac{2 (a!)^2 }{(a+j)! (a-j)! j^2} \\
&< \sum_{j=-\infty, j \neq 0}^\infty \frac{2 }{j^2} \\
&= \frac{2 \pi^2}{3},
\end{aligned}
$$
where in the second step we used the fact that $(a + j)! (a - j)! \ge (a!)^2$ for $|j| \le a$ when $a \ge 1$, and extended the sum over $j$ up to infinity.
\end{proof}

\begin{theorem}\label{thm:errbd}
Let $\psi(x)\in \mathbb{C}^{2a+1}$ on $x\in \mathbb{R}$ for $a\in \mathbb{Z}_{+}$. Then the error in the $(2a+1)$-point centered difference formula for the second derivative of $\psi(x)$ evaluated on a uniform mesh with spacing $h$ is at most
\begin{equation}
| O_{2a+1} | \le \frac{\pi^{3/2}}{9} e^{2a[ 1- \ln 2]} h^{2a-1} \max_x \left| \psi^{(2a+1)} (x) \right|.
\end{equation}
\end{theorem}
\begin{proof}
Using the expression for the error in Corollary 2.2 of Ref.~\cite{li2005general} and the triangle inequality we have that
\begin{align}
| O_{2a+1} | &\le \frac{h^{2a-1}}{(2a+1)!} \max_x \left| \psi^{(2a+1)} (x)\right| \sum_{j=-a, j\ne 0}^a |d_{2a+1,j}| |j|^{2a+1}\nonumber\\
&\le \frac{h^{2a-1}a^{2a+1}}{(2a+1)!} \max_x \left| \psi^{(2a+1)} (x) \right| \sum_{j=-a, j\ne 0}^a |d_{2a+1,j}| \nonumber\\
&< \frac{2\pi^2h^{2a-1}a^{2a+1}}{3(2a+1)!} \max_x \left| \psi^{(2a+1)} (x) \right|,\label{eq:O2abd}
\end{align}
where we used \lem{dcoeffbound} in the final step. 
Using Stirling's approximation and the fact that $a\in \mathbb{Z}_+$ and hence $a\ge 1$ we have that
\begin{equation}
\frac{a^{2a+1}}{(2a+1)!} \le \frac{\sqrt{a}e^{2a[1-\ln 2]}}{2(2a+1)\sqrt{\pi}} \le\frac{e^{2a[1-\ln 2]}}{6\sqrt{\pi}}  ,
\end{equation}
and then we find by substituting the result into~\eq{O2abd} that
\begin{align}
| O_{2a+1} | & \le \frac{\pi^{3/2}}{9} e^{2a[ 1- \ln 2]} h^{2a-1} \max_x \left| \psi^{(2a+1)} (x) \right|.
\end{align}
\end{proof}

We approximate the kinetic energy operator $T = - \sum_{i,n} \frac{1}{2m_i} \partial_{i,n}^2$ using this finite difference formula. By choosing $a$ sufficiently large, we can do this to arbitrary precision assuming $\psi(x)$ is smooth and that its derivatives, $\partial_x^p \psi(x)$, grow at most exponentially in magnitude with $p$.

%\finitedifferenceapprox*

%We will determine the order of the finite difference approximation $2a+1$ necessary to guarantee that the error $| (T - \tilde T)\psi(x) |$ is less than $\epsilon$ for arbitrary $\epsilon > 0$ in %\thm{T}. With these simplifications of the Hamiltonian in hand, though, let us return to the problem of simulation.

%\begin{lemma}\label{lem:abound} If in addition to the assumptions of \thm{errbd} we assume that there exists real numbers $\Lambda > 0$, $\alpha>0$, $\psi_{\max}\ge \max_{x} |\psi(x)|$ such that  $|\psi^{(2p+1)}| \le \alpha\Lambda^{2p+1} \psi_{\max}$ for any positive integer $p$ and $ h < e^{-[1-\ln 2]}/\Lambda$  then there exists a positive integer $a$ such that $| (T - \tilde T)\psi(x) |\le \epsilon$ for $a \in \Theta\left(\log \left( \frac{\eta D \Lambda}{m \epsilon h}\alpha\max_x |\psi(x)|\right) \right)$.\end{lemma}
%\begin{proof} If we substitute these assumptions into \thm{errbd} and demand that the error be at most $\epsilon$ we find \begin{equation} | (T - \tilde T)\psi(x) | \le \frac{\pi^{3/2} \eta D }{18m} e^{2a[ 1- \ln 2]} \gamma^{2a-1} \Lambda^{2}\alpha\psi_{\max} = \epsilon.\label{eq:O2ap1} \end{equation} Solving for $a$ yields \begin{equation} \label{eq:aquadratic} a = \frac{\ln\left(\frac{\pi^{3/2} \eta D \Lambda}{18m \epsilon h}\alpha\psi_{\max}  \right)}{2(\ln(2/h\Lambda) - 1)} \in O\left(\log \left( \frac{\eta D \Lambda}{m \epsilon h}\alpha\psi_{\max} \right) \right). \end{equation} The assumption $h < e^{-[1-\ln 2]}/\Lambda$ ensures that $a$ is well-defined and positive.  \end{proof}

\section{Applying the Hamiltonian}
\label{sec:applying}

%In this section, we will discuss the problem of applying terms from the decomposition of the kinetic and potential energy operators. Because the register in which we store the positions of the %particles, $\ket x$, is discrete, we cannot actually calculate $V(x)$, and must instead work with a discretized version of the Hamiltonian:
%\grid*

We approximate the discretized Hamiltonian $\tilde H$ by a linear combination of easily-applied unitary operators as in \eq{decomp}. To approximate $\tilde H$ to arbitrary precision while keeping the operators $V_\chi$ simple (i.e.~only signature matrices and adders), we will in fact approximate the scaled Hamiltonian $M\tilde H$ by a linear combination of unitary operators, i.e.,
\be
\label{eq:tildedecomp}
M \tilde H \approx \sum_\chi d_\chi V_\chi,
\ee
where $M>0$ determines the precision to which the sum, divided by $M$, approximates $\tilde H$. Then, rather than simulating evolution under the Hamiltonian $\tilde H$ for time $t$, we instead simulate evolution of $\psi(y(x))$ under the scaled Hamiltonian $M\tilde H$ for time $t / M$. In \sec{simulation}, we discussed how to simulate evolution when the Hamiltonian is a linear combination of unitary operators using two operators: $\select{W}$ (\eq{selectW}), which maps $\ket{\alpha} \ket\psi \mapsto \ket{\alpha} W_\alpha \ket\psi$, where $\sum_\alpha c_\alpha W_\alpha$ approximates $U(t/r)$, and $B$ (\eq{B}), which maps $\ket0 \mapsto \frac{1}{\sqrt{c}} \sum_\alpha \sqrt{c_\alpha} \ket{\alpha}$,
where $c = \sum_\alpha c_\alpha$ is the normalization constant $\sqrt{c}$ squared \cite{berry2015simulating}.

As in Ref.~\cite{berry2015simulating}, we construct $\select{W}$ using $K$ copies of an operator $\select{V}$ which chooses a single unitary operator $V_\chi$ from the sum \eq{tildedecomp} and applies it to the position register $\ket{x}$. In this section, we construct the operator $\select{V}$ so that it determines $V_\chi$ using an index register $\ket\chi$. Its action is
\be
\label{eq:selectVsimple}
\select{V} \ket\chi \ket{x} = \ket\chi V_\chi \ket{x}.
\ee

Let us explain \eq{selectVsimple} in greater detail. We wish to approximately apply $M\tilde H = M(\tilde T + \tilde V)$ to $\ket{x}$. However, $M\tilde H$ is not in general unitary, so we approximate $M\tilde H$ by a linear combination of unitary operators $V_\chi$ to some precision, and then use products of that linear combination to simulate evolution under $M\tilde H$. This means that for simulation we must work with a superposition of the different states $\chi$ of the index register, weighted by the factors $d_\chi$.

We show below that $M\tilde H$ can be approximated to arbitrary precision $\delta > 0$ by a linear combination of unitary operators---specifically, adder circuits and signature matrices---with $|\{\chi\}| = 2\eta Da + \left\lceil V_{\max} / \delta \right\rceil$ terms for a general potential bounded by $V_{\max}$. For the modified Coulomb potential this can be done with $2\eta Da + \left\lceil \frac{\eta (\eta - 1) q^2}{2 \Delta \delta} \right\rceil$ terms.

\begin{lemma}
\label{lem:selectVerrbd}
Let $V$ be some position-dependent potential energy operator bounded by $\| V(x) \|_{\infty} \le V_{\max}$, where $V_{\max} \ge 0$. Let $\delta > 0$, and let $\tilde H$ be the Hamiltonian in \defn{grid}, with the discretized potential energy operator $\tilde V(x) = V(y(x))$. We can approximate $\tilde H$ to accuracy $\delta$ by a linear combination of $2\eta D a$ addition circuits and $M \ge V_{\max} / \delta$ signature matrices, that is,
$$
\left\| \tilde H - \frac{1}{M} \sum_{\chi} d_\chi V_{\chi} \right\| \le \delta,$$
where $d_\chi > 0$ and each $V_\chi$ is either a unitary adder or a signature matrix.
\end{lemma}
\begin{proof}
$\tilde T$ is purely off-diagonal, and $\tilde V$ is purely diagonal. Furthermore, $\tilde T$ is a sum of the finite difference operators $S_{i,n}$ of \defn{dcoeffs}. From the central difference formula \eq{fdf},
\be
\label{eq:Tchis}
M \tilde T = M h^{-2} \sum_{i,n} \sum_{j=-a, j\neq 0}^{j=a} \frac{d_{2a+1,j}}{2m_i} A_j,
\ee
where $A_j$ represents unitary addition by $j$, $A_j \ket{x_{i,n}} = A_j \ket{ (x_{i,n} + j) \!\!\!\mod b }$. Because $\tilde V$ is purely diagonal, we can approximate $M\tilde V$ to precision $V_{\max}$ by
\be
\label{eq:Vchis}
M\tilde V \approx V_{\max} \sum_{j=1}^M S_j,
\ee
where each $S_j$ is a signature matrix (a diagonal matrix whose elements are all $\pm1$). With these decompositions, the Hamiltonian $M\tilde H$ is explicitly a linear combination of addition circuits and signature matrices. 

Let $d_\chi$ and $V_\chi$ be defined as in \eq{Tchis} and \eq{Vchis}: for $0 \le \chi < 2\eta D a$, $d_\chi = M h^{-2} \frac{d_{2a+1,j}}{2m_i}$ and $V_\chi = A_j$, and for $2\eta D a \le \chi < 2\eta D a + M$, $d_\chi = V_{\max}$ and $V_\chi = S_j$. (We do not specify exactly the mapping between $\chi$ and $(i, n, j)$.)
We consider the error in the diagonal of $\tilde H$, and then the error in the off-diagonal. The diagonal of $M \tilde H$ is $M\tilde V$; it is approximated by the sum $\sum_{\chi \ge 2\eta D a} d_\chi V_\chi$. The off-diagonal matrix elements of $M\tilde H$ are $M\tilde T$ and are given exactly by the sum $\sum_{\chi < 2\eta D a} d_\chi V_\chi$, so the error in the off-diagonal is zero.
Thus, the error in approximating $M\tilde H$ is only the error in approximating $M\tilde V$. As in \eq{Vchis}, the error in approximating $M\tilde V$ is at most $V_{\max}$, so the error in approximating $\tilde V$, and hence $\tilde H$, is at most $V_{\max} / M$. Choosing $M \ge V_{\max} / \delta$ then ensures that 
\be
\left\|\tilde H -\frac 1M \sum_\chi d_\chi V_\chi\right\|_{\max}\le \delta.
\ee
Finally our result follows from the fact that the max--norm and the spectral norm are equal for diagonal operators.
\end{proof}

The modified Coulomb potential energy operator of \eq{simppot} has $\| V_\text{Coulomb} \|_{\infty} \le V_{\max} = \frac{\eta (\eta - 1) q^2}{2 \Delta}$, which implies that its discretized counterpart $\tilde V$ also satisfies $\| \tilde V \|_{\infty} \le \frac{\eta (\eta - 1) q^2}{2 \Delta}$. Hence $M = \left\lceil \frac{\eta (\eta - 1) q^2}{2 \Delta \delta} \right\rceil \ge \lceil V_{\max} / \delta \rceil$ is sufficient by \lem{selectVerrbd}, and we can approximate $\tilde H$ with the potential $\tilde V(x) = V_\text{Coulomb}(y(x))$ to precision $\delta$ by a linear combination of $2\eta Da + \left\lceil \frac{\eta (\eta - 1) q^2}{2 \Delta \delta} \right\rceil$ terms. The index register $\ket\chi$ must determine which of the $2\eta D a + M$ unitary operators to apply, so it must have at least $\lceil \log(2\eta D a + M ) \rceil$ qubits.

\section{Complexity of evolving under the Hamiltonian}
\label{sec:evolving}

We now analyze the complexity of evolving under the discretized Hamiltonian using the BCCKS technique for Hamiltonian simulation \cite{berry2015simulating}, reviewed in \sec{simulation}. We break the total simulation time $t$ into $r$ segments each of length $t/r$.
Then, we approximate the time evolution operator $U(t) = \exp(-i \tilde Ht)$ by a Taylor series truncated to order $K$, which results in an error of $O\left(\frac{(\| \tilde H \| t/r)^{K+1}}{(K+1)!}\right)$. By choosing $r \ge \|\tilde H\|t$, the numerator of the leading error term is less than or equal to 1; we can thus choose $K \in O\left(\log(r/\epsilon)/\log\log(r/\epsilon)\right)$ to bound the Taylor series simulation error to $O(\epsilon)$ \cite{berry2015simulating}.

The final two missing pieces from the simulation are the operators $B$ and $\select{W}$ from \sec{simulation}. We wish to simulate time evolution under $\tilde H$ for time $t/r$ (equivalently, under $M \tilde H$ for time $t / rM$), which we can do approximately as in \eq{tts_chiform},
$$
%\label{eq:tts_chiform2}
U (t/r) \approx W(t / r) = \sum_{k=0}^K \sum_{\chi_1, \dotsc, \chi_k} \frac{(-it/rM)^k}{k!} d_{\chi_1} \dotsm d_{\chi_k} V_{\chi_1} \dotsm V_{\chi_k},
$$
using $U(t) = \exp \big(-i\tilde H t \big) = \exp \big(-i M\tilde H (t / M) \big)$. 
In \eq{alpha} of \sec{simulation}, we defined a multi-index $\alpha$ encompassing $k$ and $\chi_1$ through $\chi_k$, through which we simplified this expression to
$$
W(t/r) = \sum_{\alpha} c_\alpha W_\alpha,
$$
where $c_\alpha = (t / rM)^k  / k! \prod_{k^\prime=1}^k d_{\chi_{k^\prime}}$ includes all the real coefficients, and $W_\alpha = (-i)^k \prod_{k^\prime=1}^kV_{\chi_k}$ includes all the unitary operators. Recall from \eq{selectW} that $\select{W}$ gives $W_\alpha$, and from \eq{B} that $B$ handles the appropriate coefficients and the sum over the index register $\ket \chi$, that is, $c_\alpha$. The implementation of $B$ is discussed in Ref.~\cite{berry2014exponential}, and is of less interest because of our cost model.  Let us describe how to implement $\select{W}$ in detail.

The operator $\select{W}$ is $K$ controlled applications of $\select{V}$ and $K$ controlled phase gates.  The operators $V_{\chi_i}$ are those obtained by applying $\select{V}$ to the index state $\ket{\chi_i}$. We can obtain a product of up to $K$ of these operators by using $K$ copies of the index registers $\ket\chi$, and applying $\select{V}$ to each of them. 
But how do we account for the fact that we do not always want a product of exactly $K$ operators $-iV_{\chi_i}$? This is done using a register $\ket{k}$ of $K$ qubits, which encodes the value $k$ in unary. We apply $\select{V}$ to the index register $\ket{\chi_i}$, \emph{controlled} on the $i^\text{th}$ qubit of $\ket k$. 
We can apply the phase gates directly to the qubits of $\ket k$; these gates need not be controlled. The unary register $\ket k$, as well as the $K$ index registers $\ket{\chi_1} \dotsm \ket{\chi_K}$, are initialized in some superposition state by $B$.

As in \sec{simulation} and Ref.~\cite{berry2015simulating}, the action of the operator $A = (B^\dagger \otimes \mathds1) \select{W} (B \otimes \mathds1)$ on the ancilla and state registers is given by
\be
A \ket{0} \ket\psi = \frac{1}{c} \ket{0} W(t/r) \ket\psi + \sqrt{1 - \frac{1}{c^2}} \ket\phi,
\ee
where $\ket\phi$ is some state with the ancilla orthogonal to $\ket{0}$. Provided that $c \approx 2$, we can perform oblivious amplitude amplification: with $P_0$ the projection operator onto the zero ancilla state and $R = \mathds1 - 2P_0$, the oblivious amplitude amplification operator $G = -A R A^\dagger R A$ satisfies
\be
\label{eq:P0G}
\| P_0 G \ket{0} \ket\psi - \ket{0} U(t/r) \ket\psi \| \in O(\epsilon / r).
\ee
We repeat this process $r$ times to approximate the action of $U(t)$ on $\ket\psi$ to precision $O(\epsilon)$.  We now show that this process can expediently simulate quantum dynamics in real space and thereby prove \thm{querybound}.

\begin{proofof}{\thm{querybound}}
Rather than simulating $\tilde H$ acting on $\psi(y(x))$ for time $t$, we instead simulate the approximation of $M \tilde H$ from \lem{selectVerrbd} for time $t / M$, where $M \ge \lceil V_{\max} / \delta \rceil$. The error in approximating $\tilde H$ is $\delta$, so we must choose $\delta \in O(\epsilon / t)$ so that $\left\| e^{-i \tilde H t} - e^{-i (t / M) \sum_{\chi} d_\chi V_{\chi}} \right\| \le O(\epsilon)$. The Hamiltonian is simulated by applying $P_0 G$ (\eq{P0G}) $r$ times. This simulates evolution under $M\tilde H \approx \sum_{\chi} d_\chi V_{\chi}$ for time $t / M$ to precision $O(\epsilon)$; by the triangle inequality, the total precision is also $O(\epsilon)$.

Each application of $P_0 G$ uses $A$ three times and each application of $A$ uses $\select{W}$ once. $\select{W}$ applies a product of up to $K$ unitary operators $V_\chi$ and uses $\select{V}$ $K$ times, where 
\be
\label{eq:Kscale}
K \in O\left( \frac{\log(r / \epsilon)}{\log\log(r / \epsilon)}\right),
\ee
 from Eq.~(4) of Ref.~\cite{berry2015simulating}.  Therefore the cost of applying $P_0 G$ is $K$ times the query complexity of implementing $\select{V}$.

At first glance $\select{V}$ would seem to require $2\eta D a+M$ queries because the Hamiltonian can be decomposed into $2\eta D a +M$ unitary matrices; in fact, it can be implemented using $\Theta(1)$ queries.  To see this, let us first begin with implementing $V(x)$.  We see from the arguments of Ref.~\cite{berry2015simulating} that such a term can be simulated using a single query to an oracle that gives $V(x)$ and a polynomial amount of additional control logic to make the $M$ unitary terms (known as signature matrices) sum to the correct value.  The less obvious fact is that the $2\eta D a$ unitary adders can also be simulated using similar intuition.  This can be seen by noting that using appropriate control logic, it is possible to swap the register that a given adder acts on to a common location so only one addition needs to be performed.

To see this, consider the register that stores the index for the state $\ket{\chi} = \ket{i,n,j}$ where $n$ is the index for the dimension, $i$ is the index for the particle, and $j\in [-a,\dotsc, a]\setminus 0$.  Consider a computational basis state that encodes the positions of each particle and the unitary adder that needs to be performed of the form 
\be
\label{eq:orig_reg}
\ket{\chi} \ket{x_{1,1}}\dotsm\ket{x_{i,n}}\dotsm \ket{x_{\eta,D}}.
\ee
Then using the data in $\ket{\chi}$ a series of swap operations can be performed such that
\be
\label{eq:swap_reg}
\ket{\chi} \ket{x_{1,1}}\dotsm\ket{x_{i,n}}\dotsm \ket{x_{\eta,D}} \mapsto \ket{\chi} \ket{x_{i,n}} \dotsm \ket{x_{1,1}} \dotsm \ket{x_{\eta,D}}.
\ee
The addition of $j$ to this register can then be performed by applying,
\be
\label{eq:add_reg}
\ket{\chi} \ket{x_{i,n}} \dotsm \ket{x_{1,1}} \dotsm \ket{x_{\eta,D}} \mapsto \ket{i,n}{\rm Add}\left(\ket{j} \ket{x_{i,n}}\right) \dotsm \ket{x_{1,1}} \dotsm \ket{x_{\eta,D}}.
\ee
The desired result then follows from inverting the swap gates.  Since quantum mechanics is linear, any unitary that performs such swaps for an arbitrary value of $\chi$ that is stored in the ancilla register will also have the correct action on a superposition state.  Thus it is possible to perform the addition using a single query to an adder circuit, given that such a network of controlled swaps can be implemented.

\begin{figure}[t!]
\begin{center}
\includegraphics[width=0.37\textwidth, trim={2cm 4.35cm 18cm 3cm},clip]{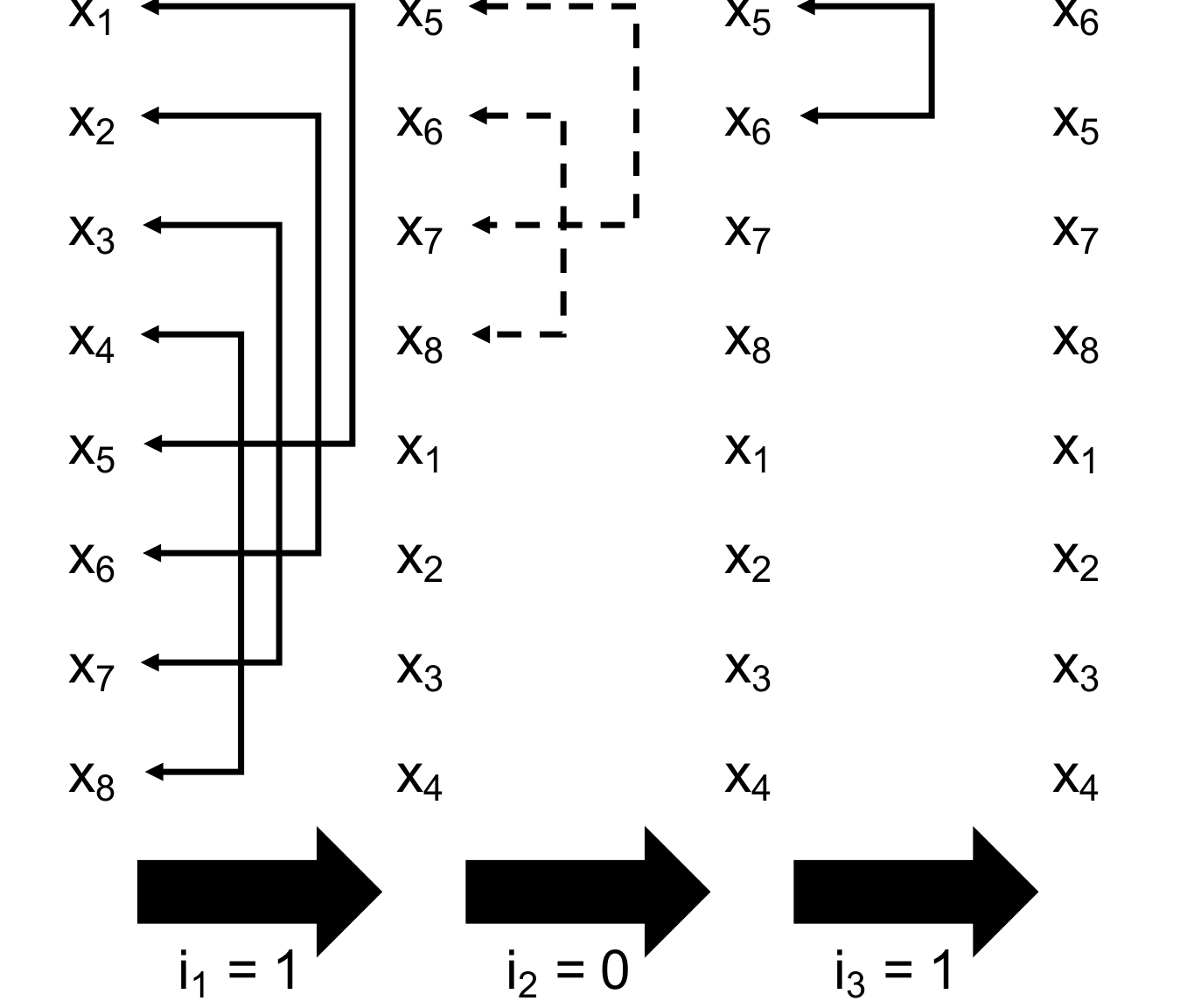}
\caption{The controlled swap procedure used in the proof of \thm{querybound} for $i=6$ with $\eta = 8$. $i=6$ is stored in a register as $101$. Since the first bit of $i$ is 1, $\ket{x_{i^\prime}}$ is swapped with $\ket{x_{i^\prime - 4}}$ for $i^\prime \in [5, 8]$ (solid arrows). The second stage (dashed arrows) is not performed since the second bit of $i$ is 0. Finally, since the third bit of $i$ is 1, $x_5$ and $x_6$ are swapped to leave $x_6$ in the first position. Though the gate count scales linearly in $\eta D$, the circuit depth is only logarithmic in it.}
\label{fig:log_cswaps}
\end{center}
\end{figure}

Such a series of swaps can be shown constructively to exist by using a strategy similar to binary search.
The steps are as follows. For $k \in [1, \lceil \log\eta \rceil]$: controlled on the $k^\text{th}$ qubit of $i$, swap $\ket{x_{i^\prime}}$ with $\ket{x_{i^\prime - 2^{\lceil \log\eta \rceil - k}}}$ for $i^\prime \in [2^{\lceil \log\eta \rceil - k} + 1, 2^{\lceil \log\eta \rceil - k + 1} ]$. \fig{log_cswaps} gives an example of this procedure for $i=6$ with $\eta = 8$.
After each stage, the desired $\ket{x_{i}}$ register is in position $i \!\!\mod 2^{\lceil \log\eta \rceil - k}$. Thus, after all $\lceil \log\eta \rceil$ iterations $\ket{x_{i}}$ occupies the first particle position $i^\prime = 1$. We repeat the same process for the coordinate $n$, so that $\ket{x_{in}}$ is first in the $\ket{x}$ register, and apply the unitary adder to it, adding by $j$. Finally, we run the sequence of controlled swap gates in reverse to return all the registers in $\ket{x}$ to their original positions. The controlled swaps require $O(\eta D \log(L/h))$ gates but only depth $O(\log(\eta D))$. 

So $\select{V}$ requires $\Theta(1)$ queries to the potential energy oracle and $\Theta(1)$ unitary adders. Each application of $P_0 G$ requires, from~\eq{Kscale}, $3K \in O\left( \frac{\log(r / \epsilon)}{\log\log(r / \epsilon)}\right)$ uses of $\select{V}$, and as such $\Theta(K)$ unitary adders and calls to the potential energy oracle.  The query complexity within our model then scales as~\cite{berry2015simulating}
\be
\Theta(Kr) \subseteq O\left(\frac{r\log(r / \epsilon)}{\log\log(r / \epsilon)} \right).
\ee
The results in Ref.~\cite{berry2015simulating}  require that $r\ge \| \tilde H \| t = \| \tilde H \| t $. A stricter requirement, that $r = \sum_\chi |d_\chi| t / \ln 2$, is given by the condition of oblivious amplitude amplification that $c = \sum_\alpha c_\alpha \approx 2$. If $r$ is not an integer, we can take the ceiling of this as $r$, and the final segment will have $c < 2$, which can be compensated for using an ancilla qubit \cite{berry2015simulating}.  In order to guarantee that we have enough segments to satisfy these requirements we choose
\be
\begin{aligned}
r =& \left\lceil \sum_\chi | d_\chi | t / M\ln 2 \right\rceil \\
%\le& ~\lceil(\|\tilde T\| + \|\tilde V\|)t/ M\ln2)\rceil \\ 
\le& \left( h^{-2} \sum_{i,n} \sum_{j=-a, j\neq 0}^{j=a} \frac{| d_{2a+1,j} |}{2m_i} + V_{\max} \right) t / \ln 2 +1,\\
\end{aligned}
\label{eq:mistake}
\ee
where we used \eq{Tchis} and \eq{Vchis}.
This upper bound on $r$ then allows us to determine an upper bound on how many adder circuits or how many queries to an oracle for the potential energy are required for simulation.

We then see from~\lem{dcoeffbound} and~\eq{mistake} that the number of times that $P_0G$ is applied, $r$, obeys
\be
r \le \left( \frac{\pi^2 \eta D }{3m h^2} + V_{\max} \right) t / \ln2+1.\label{eq:req2}
\ee
Finally using~\eq{req2} we have that the total number of queries made to  $V$ and $T$ scales as
\begin{equation}
\Theta(Kr)  \subseteq O\left( \left( \frac{\eta D}{mh^2} + V_{\max}\right)t \left[\frac{\log\left( \frac{\eta Dt}{mh^2 \epsilon} + \frac{V_{\max}t}{\epsilon}\right) }{\log\left(\log\left( \frac{\eta Dt}{mh^2 \epsilon} + \frac{V_{\max}t}{\epsilon}\right)\right)}\right] \right),
\end{equation}
as claimed.
\end{proofof}

This shows that if a modest value of $h$ can be tolerated then the continuous-variable simulation that we discuss above will require a number of resources that scales slightly superlinearly with the number of particles.  

A possible criticism of the above cost analysis is that the potential energy oracle considered requires a number of operations that scales polynomially with the number of particles were we to implement it using elementary operations for pairwise Hamiltonians such as the Coulomb Hamiltonian.  One way to deal with this is to use oracles that have complexity that is constant in the size of the simulation, such as an oracle for each of the pairwise interactions.  We show in the corollary below that switching to such a pairwise oracle and optimizing the simulation against it leads to a query complexity that is the same as that in \thm{querybound} (potentially up to logarithmic factors).

\begin{corollary}
Let $V_{ij}$ be the potential energy operator for the two-particle interaction between particles $i$ and $j$. 
With $\tilde H$ be as in \thm{querybound} with $V = \sum_{i < j} V_{ij}$, we can simulate $e^{-i\tilde H t}\ket{\psi}$ for time $t > 0$ within error $\epsilon > 0$ with 
$$\tilde O\left( \left( \frac{\eta D}{mh^2} + V_{\max}\right)t \log\left( 1/ \epsilon \right)  \right)$$ % CHECK
 unitary adders and queries to an oracle for the two-particle potential energy $V_{ij}$.
\end{corollary}
\begin{proof}
The intuition behind our approach is to use the result in~\thm{querybound} for truncated Taylor-series simulation of the particle system, but to multiply the cost of the simulation by the cost of implementing the query using the pairwise oracles.  Since there are $\eta^2$ such terms one would expect that the complexity should be $\eta^2$ times that quoted in~\thm{querybound}.  However, we can optimize the algorithm for the pairwise oracle to perform the simulation by using a swap network similar to that exploited for the kinetic energy to reduce the cost.

We replace the potential energy operator $\tilde V$ by a sum of two-particle potential energies, so that the Hamiltonian we simulate is
$$
\tilde H = h^{-2} \sum_{i,n} \sum_{j=-a, j\neq 0}^{j=a} \frac{d_{2a+1,j}}{2m_i} A_j + \sum_{i \neq j} \tilde V_{ij}.
$$
In \thm{querybound}, we showed how to implement the $2\eta D a$ terms in the kinetic energy operator using a single adder circuit, and $V$ using a single query to the total potential energy. That is, for a two-particle potential, we evaluate $\sum_{i \neq j} V_{ij}$ by a single query to $V$. Thus, in order to show our claim that we can perform a single segment of evolution under $\tilde H$ using a constant number of unitary adders and queries to an oracle for the two-particle potential energy $V_{ij}$, we must show that the potential $V_{ij}$ can be evaluated with a constant number of queries to the pairwise potential.

In general, the pairwise potential energy is a function of the properties of the particles $i$ and $j$ as well as their positions. The action of the pairwise oracle $V_p$ is 
$$
V_p \left( \ket{ij} \ket0 \ket{x_a} \ket{x_b} \right) \coloneqq \ket{ij} \ket{V_{ij}(x_a, x_b)} \ket{x_a} \ket{x_b},
$$
where $x_a$ and $x_b$ are the positions of particles $a$ and $b$. 

The implementation of a segment in the truncated Taylor series simulation requires that we implement the Hamiltonian as a linear combination of unitaries.  We showed in \thm{querybound} that the kinetic energy part of the linear combination can be implemented using a constant number of adder circuits.  Therefore, in order to show that the pairwise Hamiltonian can also be implemented using a constant number of queries in this model we need to show that the two-particle potential terms in the linear combination can be enacted using a constant number of queries to $V_p$.

We show that the potential terms can be performed using a single query to $V_p$ using a swap network reminiscent of that used for the kinetic energy terms in~\thm{querybound}. 
Let us assume that we want to implement the $\chi^{\rm th}$ term in the decomposition, $H_\chi=V_{ij}$.  Then we can write the state of the control register and the simulator subspace as
\be
\ket\chi \ket0 \ket{x_1} \dotsm \ket{x_i} \dotsm \ket{x_j} \dotsm \ket{x_\eta}.
\ee
We use the data in the control register $\ket\chi$ to perform a series of controlled-swap operations such that
\be
\ket{\chi} \ket0 \ket{x_1} \dotsm \ket{x_i} \ket{x_j} \ket{x_\eta} \mapsto \ket{\chi} \ket0 \ket{x_{i}} \ket{x_{j}} \dotsm \ket{x_{1}} \dotsm \ket{x_{\eta}}.
\ee
This process uses ${\rm poly}(\eta)$ controlled swaps and no queries.
We then query the pairwise oracle $V_p$ to prepare the state
$$
\ket{\chi} \ket0 \ket{x_{i}} \ket{x_{j}} \dotsm \ket{x_{1}} \dotsm \ket{x_{\eta}} \mapsto \ket{\chi} \ket{V_{ij}(x_i, x_j)} \ket{x_{i}} \ket{x_{j}} \dotsm \ket{x_{1}} \dotsm \ket{x_{\eta}}.
$$
Then, using the signature matrix trick, we can implement these terms as a sum of unitary operations within arbitrarily small error after appropriately cleaning the ancilla qubits.  Because this circuit works uniformly for all pairwise interactions, the entire segment can be implemented using only one application of the above routine for simulating the potential terms and the routine for simulating the kinetic terms from~\thm{querybound}.  As argued, the routine requires only a constant number of queries, and therefore each segment requires only a constant number of queries to the adder circuit and $V_p$.  The corollary then follows from the bounds on the number of segments in~\thm{querybound}.
\end{proof}

This is significant because the best methods known for performing such simulations not only require the Born-Oppenheimer approximation, but also require $\tilde{O}(\eta^5)$ operations (assuming $\eta$ is proportional to the number of spin-orbitals) \cite{babbush2016exponentially,babbush2015exponentially}.  Thus, depending on the value of $h$ needed, this approach can potentially have major advantages in simulation time.

The value of $h$ needed for such a simulation is difficult to address as it depends sensitively on the input state being simulated.  In the next section, we provide estimates of the scaling of this parameter that show that the above intuition may not hold without strong assumptions about the states being simulated.  Specifically, we find that that the value of $h$ needed to guarantee that the simulation error is within $\epsilon$ can shrink exponentially with $\eta D$ in some pathological cases.

\section{Errors in Hamiltonian model}
\label{sec:errors}

In our discussion thus far, we have introduced several approximations and simplifications of the Hamiltonian so as to make the simulation problem well-defined and also tractable. In this section, we bound the errors incurred by these choices. 
At the heart of these approximations is the discretization of the system coordinates into $b$ hypercubes of side length $h$ along each spatial direction from \defn{grid}.

We begin by bounding the errors in the kinetic and potential energy operators, starting off with an upper bound on derivatives of the wave function assuming a maximum momentum $k_{\max}$ in \lem{err1}. We apply this upper bound to determine the maximum error in the finite difference approximation for the kinetic energy operator in \thm{T}.
%Using this maximum error in the finite difference approximation for the kinetic energy operator, we determine in \cor{looseabd} the finite difference order required to suppress the error to arbitrary accuracy. 
Following that, in \lem{V}, we upper bound the error in the potential energy operator due to discretization (the difference between $V(x)$ and $\tilde V(x) = V(y(x))$ of \defn{grid}).

We shift our focus from errors in the operators to simulation errors beginning in \lem{HHtildeErr}, where we give the error in evolving under $\tilde H$ rather than $H$. In \lem{psitilde} we bound the error in evolving the discretized wave function rather than the wave function itself. We give the total simulation error in \cor{total}, and in \lem{normalization} give the difference between simulating the wave function $\psi(x)$ and the discretized wave function $\psi(y(x))$ due to normalization. Finally, in \thm{totalab} we determine the values of $a$ and $h$ needed to bound the total simulation error to arbitrary $\epsilon > 0$ in the worst case, before discussing for which states the worst case holds, and then determining the requirements on $a$ and $h$ under more optimistic assumptions about the scaling of the derivatives of the wave function in \cor{optimistic}.

We begin by introducing a lemma which we use to bound the errors in the kinetic and potential energy operators, assuming a maximum momentum:

\begin{lemma}
\label{lem:err1}
Let $\psi(k):\mathbb{R}^N \mapsto \mathbb{C}$ and $\psi(x):\mathbb{R}^N \mapsto \mathbb{C}$ be conjugate momentum and position representations of the same wave function in $N$ dimensions and assume that $\psi(k)=0$ if $\| k \|_{\infty} > k_{\rm max}$.  Then for any position component $x_{i}$, and any non-negative integer $r$,
$$| \partial_{x_{i}}^r \psi(x) | \le \frac{k_{\max}^r}{\sqrt{2r+1}}\left(\frac{k_{\max}}{\pi}\right)^{N/2}. $$
\end{lemma}
\begin{proof}
\begin{align}
\partial_{x_{i}}^r \psi(x) &=  i^r\bra{x} p_{i}^r \ket{\psi}\nn
&= i^r \int_{-\infty}^\infty \dotsi \int_{-\infty}^\infty \bra{x} p_{i}^r \ket{k}\braket{k}{\psi} \,\mathrm{d}^N k.
\end{align}
Using the momentum cutoff and the fact that $\braket{x}{k} = \frac{e^{i x\cdot k} }{ (2\pi)^{N/2} }$ in $N$ dimensions, we then have
\begin{align}
\partial_{x_{i}}^r \psi(x) 
&=\frac{ i^r}{(2\pi)^{N/2}} \int_{-k_{\max}}^{k_{\max}} \dotsi \int_{-k_{\max}}^{k_{\max}} k_{i}^r e^{i k \cdot x} \braket{k}{\psi} \,\mathrm{d}^N k.
\end{align}
Here $k_{i}$ refers to the $i^\text{th}$ component of the $k$-vector.
We then use the Cauchy-Schwarz inequality to separate the terms in the integrand to find
\begin{align}
| \partial_{x_{i}}^r \psi(x) | &\le  \frac{1}{(2\pi)^{N/2}}  \sqrt{\int_{-k_{\max}}^{k_{\max}} \dotsi \int_{-k_{\max}}^{k_{\max}} k_{i}^{2r} \,\mathrm{d}^N k \int_{-k_{\max}}^{k_{\max}}\dotsi \int_{-k_{\max}}^{k_{\max}} |\psi(k)|^2 \,\mathrm{d}^N k}\nn
&= \frac{1}{(2\pi)^{N/2}}\sqrt{\int_{-k_{\max}}^{k_{\max}} \dotsi \int_{-k_{\max}}^{k_{\max}} k_{i}^{2r} \,\mathrm{d}^N k}\nn
&= \frac{k_{\max}^r}{\sqrt{2r+1}}\left(\frac{k_{\max}}{\pi} \right)^{N/2}. % the integral is $\sqrt{ k_{\max}^{2r} \times (2k_{\max})^{N]) }$
\end{align}
\end{proof}

Recall that $m = \min_i m_i$ is the minimum mass of any particle in the system. \lem{err1} leads to the following useful bounds:
\begin{align}
| \psi(x) | &\le \left(\frac{k_{\max}}{\pi} \right)^{\eta D/2}\label{eq:psibd}\\
| \partial_{x_{i,n}} \psi(x) | &\le\frac{k_{\max}}{\sqrt{3}} \left(\frac{k_{\max}}{\pi} \right)^{\eta D/2} \label{eq:diffpsieq}\\
| T\psi(x) | & \le \frac{\eta Dk_{\max}^2}{2 m \sqrt{5}}\left(\frac{k_{\max}}{\pi} \right)^{\eta D/2}\\
| T\partial_{x_{i,n}}\psi(x) | & \le \frac{\eta Dk_{\max}^3}{2 m \sqrt{7}}\left(\frac{k_{\max}}{\pi} \right)^{\eta D/2}
\end{align}

We now bound the error in the finite difference approximation for the kinetic energy operator using \lem{err1}.

\begin{theorem}
Let $\psi(k): \mathbb{R}^{\eta D} \mapsto \mathbb{C}$ and $\psi(x): \mathbb{R}^{\eta D} \mapsto \mathbb{C}$ be conjugate momentum and position representations of an $\eta$-particle wave function in $D$ spatial dimensions satisfying the assumptions of \lem{err1}, and let $T = \sum_{i,n} T_{i,n}$ and $\tilde T = \sum_{i, n} S_{i, n} $, where $T_{i,n} = p_{i,n}^2/2m_i$. Then
$$
| (T - \tilde T ) \psi(x) | \le \frac{ \pi^{3/2} e^{2a[ 1- \ln 2]}}{18m \sqrt{4a+3}} \eta D k_{\max}^{2a+1} \left(\frac{k_{\max}}{\pi} \right)^{\eta D/2} h^{2a-1} ,
 $$
where $m = \min_i m_i$.
\label{thm:T}
\end{theorem}
\begin{proof}
Recall from \thm{errbd} that, for a single coordinate, the error $| O_{2a+1} |$ in the $(2a+1)$-point central difference approximation of the second derivative is upper-bounded by
$$\frac{\pi^{3/2}}{9} e^{2a[ 1- \ln 2]} h^{2a-1} \max_x \big| \psi^{(2a+1)} (x) \big|,$$
where $h$ is the grid spacing. By \lem{err1}, $\max_x \left| \psi^{(2a+1)} (x) \right| \le \frac{k_{\max}^{2a+1}}{\sqrt{4a+3}}\left(\frac{k_{\max}}{\pi} \right)^{\eta D/2}$. Thus for any coordinate $(i, n)$,
\be
\begin{aligned}
| (T_{i,n} - S_{i,n} ) \psi(x) | &\le \frac{\pi^{3/2}}{18m} e^{2a[ 1- \ln 2]} h^{2a-1} \frac{k_{\max}^{2a+1}}{\sqrt{4a+3}}\left(\frac{k_{\max}}{\pi} \right)^{\eta D/2}.\label{eq:errbd}
\end{aligned}
\ee
The result follows by summing over all $\eta$ particles and $D$ dimensions.
\end{proof}
\thm{T} %, similar to \lem{abound},
does not require that $\psi(x)$ be discretized as in \defn{grid}: the second derivative of any wave function with maximum momentum $k_{\max}$ can be calculated in this way. We have finished addressing the error in the kinetic energy operator and move now to the error in the potential energy operator.

\begin{lemma}
\label{lem:V}
Let $\psi(x): \mathbb{R}^{\eta D} \mapsto \mathbb{C}$ and $\tilde{V}:\mathbb{R}^{\eta D}\mapsto \mathbb{R}$ satisfy the assumptions of \lem{err1} such that $\tilde{V}(x)=V(y(x))$, where $\| \nabla V(x) \|_{\infty} \le V^\prime_{\max}$. Then
$$| (V-\tilde{V})\psi(x) | \le \frac{h {\eta D} }{2} \left(\frac{k_{\max}}{\pi} \right)^{\eta D/2} V^\prime_{\max}.$$
In particular, for the modified Coulomb potential energy operator,
\be
\left| (V_\text{Coulomb} - \tilde{V}_\text{Coulomb}) \psi(x) \right| \le \frac{h {\eta D}}{2} \frac{\eta^2 q^2 \sqrt{3}}{9\Delta^2} \left(\frac{k_{\max}}{\pi} \right)^{\eta D/2}.
\ee
\end{lemma}
\begin{proof}
$V$ and $\tilde V$ differ in that $\tilde V$ is evaluated at a centroid of a hypercube whereas $V$ is evaluated at the ``true'' coordinates. The distance from the corner of a hypercube to its center is at most $h\sqrt{\eta D}/2$, so because $\| \cdot \| \le \sqrt{\eta D} \| \cdot \|_\infty$ for vectors of dimension $\eta D$,
% The total volume is $L^{\eta D}$ and there are $b^{\eta D}$ hypercubes. Each hypercube thus has volume $(L/b)^{\eta D} = h^{\eta D}$. We can be off by up to $h/2$ in each coordinate (each spatial dimension for each particle) so the distance from the corner of a hypercube to its center is at most $h \sqrt{\eta D} / 2$.
\be
| (V - \tilde{V}) \psi(x) | = | (V(x) - V(y(x)))\psi(x)| \le \frac{h {\eta D}}{2} \left| \max_i \partial_{x_i}V(x) \right| \left(\frac{k_{\max}}{\pi} \right)^{\eta D/2},
\ee
where we used the bound on $|\psi(x)|$ of \eq{psibd}. The result follows from the assumption that $\| \nabla V(x) \|_{\infty} \le V^\prime_{\max}$.

For the modified Coulomb potential energy operator $V_\text{Coulomb} = \sum_{i<j} \frac{q_i q_j}{\sqrt{\| x_i - x_j \|^2 + \Delta^2}}$ it is easy to verify that
\begin{equation}
| \partial_{x_i}V_\text{Coulomb}(x) | = \left| \partial_{x_i} \sum_{k\ne j}\frac{q_k q_j}{\sqrt{\| x_k -x_j \|^2 +\Delta^2}} \right| \le \frac{\eta^2 q^2}{2} \max \left| \partial_{x_i}\frac{1}{\sqrt{\| x_k -x_j \|^2 +\Delta^2}}\right| \le \frac{\eta^2 q^2 \sqrt{3}}{9\Delta^2} \label{eq:diffx},
\end{equation}
from which the second result follows.
\end{proof}

At this point, we have bounds on the error in the approximations of the kinetic and potential energy operators. 
We apply these to determine the error in simulating $\tilde H$ rather than $H$. After that, we determine the maximum error in time-evolving the discretized wave function $\psi(y(x))$ rather than $\psi(x)$, and then combine the two results.

\begin{lemma}
\label{lem:HHtildeErr}
If the assumptions of~\thm{T} are met for the wave functions $e^{-iH s}\psi$ and $e^{-i\tilde{H} s}\psi$ for all $s\in [0,t]$ where $\psi: \mathbb{R}^{\eta D}\mapsto \mathbb{C}$, and $| \nabla V(x)|_{\infty} \le V^\prime_{\max}$, then for any square integrable $\phi:\mathbb{R}^{\eta D}\mapsto \mathbb{C}$ and $Q\subseteq S$
$$\left|\int_Q \phi^* \left(e^{-iHt} -e^{-i\tilde H t}\right) \psi \,\mathrm{d}^{\eta D}x\right| 
\le t\left(\frac{ \pi^{3/2} e^{2a[ 1- \ln 2]}}{\sqrt{4a+3}} \frac{\eta D k_{\max}^{2a+1} h^{2a-1}}{18 m} + \frac{h {\eta D} }{2} V^\prime_{\max} \right)\left(\frac{k_{\max}}{\pi} \right)^{\eta D/2}\sqrt{\int_Q \mathrm{d}x^{\eta D}\int_Q | \phi |^2 \,\mathrm{d} x^{\eta D}}$$
\end{lemma}
\begin{proof} From the Cauchy-Schwarz inequality
\begin{equation}
\left|\int_Q \phi^* \left(e^{-iHt} -e^{-i\tilde H t}\right) \psi(x) \,\mathrm{d}^{\eta D}x\right|\le  \max_x \left| \left(e^{-iHt} -e^{-i\tilde H t}\right) \psi(x) \right| \sqrt{\int_Q \mathrm{d}x^{\eta D}\int_Q \| \phi \|^2 \,\mathrm{d} x^{\eta D}}
\end{equation}
Repeating the standard argument from Box 4.1 of Nielsen and Chuang \cite{nielsen2011quantum} and using the fact that for the input state $\psi$, $| H\psi(x,t) |$ is bounded, we have that
\begin{equation}
\left| \left(e^{-iHt} -e^{-i\tilde H t}\right) \psi(x) \right|=\lim_{r\rightarrow \infty} \left| \left(\left(e^{-iHt/r}\right)^r -\left(e^{-i\tilde H t/r}\right)^r\right) \psi(x) \right| \le \max_{x,\psi} \left| (H-\tilde{H}) \psi(x) \right| t .
\end{equation}
Here the maximization over $\psi$ is meant to be a maximization over all $\psi$ that satisfy the assumptions of~\thm{T}.

We then apply \thm{T} to find that
\begin{equation}
\max_x |(T-\tilde T) \psi(x) | \le \frac{ \pi^{3/2} e^{2a[ 1- \ln 2]}}{\sqrt{4a+3}} \frac{\eta D k_{\max}^{2a+1} h^{2a-1}}{18 m } \left(\frac{k_{\max}}{\pi} \right)^{\eta D/2}.
\end{equation}
Similarly we have from \lem{V} that
\be \max_x | (V-\tilde{V}) \psi(x) | \le \frac{h {\eta D}}{2} \left(\frac{k_{\max}}{\pi} \right)^{\eta D/2} V^\prime_{\max} .\ee
The claim of the lemma then follows by combining these three parts together.
\end{proof}

\begin{lemma}
\label{lem:psitilde}
If the assumptions of~\thm{T} are met for the wave function $e^{-i\tilde{H} s}\psi$ for all $s\in [0,t]$ where $\psi: \mathbb{R}^{\eta D}\mapsto \mathbb{C}$ then for any square integrable $\phi:\mathbb{R}^{\eta D}\mapsto \mathbb{C}$, $v\in \mathbb{R}^{\eta D}$ such that $\|v\| \le h\sqrt{\eta D}/2$ and $Q\subseteq S$
$$
\left|\int_Q \phi^*(x) (e^{-i\tilde H t}\psi(x) - e^{-i\tilde H t} \psi(x+v))\,\mathrm{d}x^{\eta D}\right| \le \frac{k_{\max} \eta D h}{2\sqrt{3}}\left(\frac{k_{\max} }{\pi} \right)^{\eta D/2} \sqrt{\int_Q \mathrm{d}x^{\eta D}\int_Q | e^{iHt} \phi |^2 \,\mathrm{d} x^{\eta D}}.
$$
\end{lemma}
\begin{proof}
Under our assumptions we have that
\begin{equation}
\left|\int_Q \phi^*(x) (e^{-i\tilde H t}\psi(x) - e^{-i\tilde H t} \psi(x+v))\,\mathrm{d}x^{\eta D}\right| = \left|\int_Q \phi^*(x)e^{-i\tilde H t} (\psi(x) - \psi(x+v))\,\mathrm{d}x^{\eta D}\right|.
\end{equation}
Since $\psi(x)$ is differentiable, we have from the fact that for vectors of dimension $\eta D$, $\| \cdot \| \le \sqrt{\eta D} \|\cdot \|_\infty$ that
\be
| \psi(x) - \psi(x+v) | \le \|v\| \max_x \|\nabla \psi(x)\| \le \frac{\eta Dh}{2} \max_x | \partial_{x_{i,n}} \psi(x) |.
\ee
\eq{diffpsieq} then implies that
\begin{equation}
\label{eq:psihcube}
|\psi(x) - \psi(x+v)| \le \frac{k_{\max} \eta D h}{2\sqrt{3}}\left(\frac{k_{\max}}{\pi} \right)^{\eta D/2}.
\end{equation}
The remainder follows from the Cauchy-Schwarz inequality.
\end{proof}

Recall from \defn{grid} that $y:x\mapsto \min_{u\in \{y_j\}} \|x-u\|$, so $\|x-y(x)\| \le h\sqrt{\eta D}/2$. \lem{psitilde} is thus slightly more general than just bounding the error in time-evolving $\psi(y(x))$ rather than $\psi(x)$, but it suffices for our purposes. We next combine the previous two lemmas to bound the error in evolving the discretized wave function $\psi(y(x))$ under the discretized Hamiltonian $\tilde H$ rather than evolving the true $\psi$ under $H$.

\begin{corollary}
\label{cor:total}
If the assumptions of~\thm{T} are met for the wave functions $e^{-iH s}\psi$ and $e^{-i\tilde{H} s}\psi$ for all $s\in [0,t]$ where $\psi: \mathbb{R}^{\eta D}\mapsto \mathbb{C}$, and $\| \nabla V(x)\|_{\infty} \le V^\prime_{\max}$, then for any square integrable wave function $\phi:S\mapsto \mathbb{C}$ we have that $|\int_S \phi^*(x) e^{-i H t} \psi(x) \,\mathrm{d}^{\eta D}x -\int_S \phi^*(x) e^{-i \tilde{H} t} \psi(y(x)) \,\mathrm{d}^{\eta D}x|$ is bounded above by
$$
\left[\frac{k_{\max} \eta D h}{2\sqrt{3}} + t \left( \frac{ \pi^{3/2} e^{2a[ 1- \ln 2]}}{\sqrt{4a+3}} \frac{\eta D k_{\max}^{2a+1} h^{2a-1}}{18 m } + \frac{h {\eta D} }{2} V^\prime_{\max} \right) \right]\left(\frac{k_{\max} L }{\pi} \right)^{\eta D/2}.
$$
\end{corollary}
\begin{proof}
By the triangle inequality,
\begin{align}
&\left|\int_S \phi^*(x) e^{-i H t} \psi(x) \,\mathrm{d}^{\eta D}x -\int_S \phi^*(x) e^{-i \tilde{H} t} \psi(y(x)) \,\mathrm{d}^{\eta D}x\right|\nn
&\qquad\le \left|\int_S \phi^*(x) e^{-i H t} \psi(x) \,\mathrm{d}^{\eta D}x -\int_S \phi^*(x) e^{-i \tilde{H} t} \psi(x) \,\mathrm{d}^{\eta D}x \right|\nn
&\qquad ~+ \left|\int_S \phi^*(x) e^{-i \tilde H t} \psi(x) \,\mathrm{d}^{\eta D}x -\int_S \phi^*(x) e^{-i \tilde{H} t} \psi(y(x)) \,\mathrm{d}^{\eta D}x \right|. \label{eq:HtildeH}
\end{align}
\lem{HHtildeErr} and \lem{psitilde} can be used to bound these terms.  First note that because we assume that $\phi$ is a wave function that has support only on $S$, it follows
from the definition of $\tilde{T}$ that $\tilde{T} \phi$ does also.  Therefore it follows from Taylor's theorem and the fact that $\tilde{V}$ is diagonal that $e^{-i \tilde H t}\phi$ has support only on $S$.  Since $\phi$ has norm $1$ this implies that
\begin{equation}
\sqrt{\int_S \mathrm{d}x^{\eta D}\int_S |e^{i\tilde H t}\phi|^2 \,\mathrm{d} x^{\eta D}} =L^{\eta D/2},
\end{equation}
and similarly
\begin{equation}
\sqrt{\int_S \mathrm{d}x^{\eta D}\int_S |\phi|^2 \,\mathrm{d} x^{\eta D}} =L^{\eta D/2}.
\end{equation}
The result then follows by substituting these results as well as those of \lem{HHtildeErr} and \lem{psitilde} into~\eq{HtildeH}.
\end{proof}

A final issue is that, while $\psi(x)$ is normalized, $\psi(y(x))$ in general will not be. Initializing the quantum computer renormalizes $\psi(y(x))$, so the wave function simulated by the quantum computer is in fact $ \psi(y(x)) \bigg/ \sqrt{\int_S  |\psi(y(x))|^2 \,\mathrm{d} x^{\eta D}}$.  The following lemma bounds the contribution of this final source of error.

\begin{lemma}\label{lem:normalization}
If the assumptions of~\lem{err1} hold then for any bounded Hermitian operator $H$, $t\ge 0$, and square integrable wave function $\phi:S \mapsto \mathbb{C}^{\eta D}$ such that $\int_S |\phi(x)|^2\mathrm{d}x^{\eta D}=1$, we have that
$$
\left| \int_S \phi(x)^* e^{-iHt} \psi(y(x)) \mathrm{d} x^{\eta D}-\int_S \frac{\phi(x)^* e^{-iHt} \psi(y(x))}{\sqrt{\int_S |\psi(y(x))|^2\mathrm{d}x^{\eta D}}} \mathrm{d} x^{\eta D}\right|\le \delta,
$$
for $$h\le 3\sqrt{\frac{\min(\delta,\sqrt{3/8})}{\eta D}}\frac{1}{k_{\max}}\left(\frac{k_{\max} L}{\pi} \right)^{-\eta D/2}. $$
\end{lemma}
\begin{proof}
The Cauchy-Schwarz inequality and the fact that $\phi(x)$ is normalized show that
\begin{align}
&\left| \int_S \phi(x)^* e^{-iHt} \psi(y(x)) \mathrm{d} x^{\eta D}-\int_S \frac{\phi(x)^* e^{-iHt} \psi(y(x))}{\sqrt{\int_S |\psi(y(x))|^2\mathrm{d}x^{\eta D}}} \mathrm{d} x^{\eta D}\right|\nn
&\qquad\le \left|\int_S \phi(x)^* e^{-iHt} \psi(y(x)) \mathrm{d} x^{\eta D} \right|\left|1-\frac{1}{\sqrt{\int_S |\psi(y(x))|^2 \mathrm{d}x^{\eta D}}} \right|.\nonumber\\
&\qquad\le \sqrt{\int_S |\psi(y(x))|^2\mathrm{d}x^{\eta D}\int_S |\phi(x)|^2\mathrm{d}x^{\eta D}}\left|1-\frac{1}{\sqrt{\int_S |\psi(y(x))|^2 \mathrm{d}x^{\eta D}}} \right|\nn
&\qquad= \sqrt{\int_S |\psi(y(x))|^2\mathrm{d}x^{\eta D}}\left|1-\frac{1}{\sqrt{\int_S |\psi(y(x))|^2 \mathrm{d}x^{\eta D}}} \right|.
\end{align}
Next, by applying the midpoint rule on each of the $\eta D$ dimensions in the integral we have that
\begin{align}
\left|\int_S |\psi(y(x))|^2 \mathrm{d}x^{\eta D}-\int_S |\psi(x)|^2 \mathrm{d}x^{\eta D}\right|&=\left|\int_S |\psi(y(x))|^2 \mathrm{d}x^{\eta D}-1\right|\nn
&\le \frac{\eta D h^2 \max \big| \partial^2_{x_{i,n}}|\psi(x)|^2 \big| L^{\eta D}}{24}.\label{eq:intbd72}
\end{align}
Using the fact that $|\psi(x)|^2 = \psi(x)\psi^*(x)$ we find that
\begin{equation}
\max \big| \partial^2_{x_{i,n}}|\psi(x)|^2 \big| \le 2 \max |\partial^2_{x_{i,n}} \psi(x)| \max|\psi(x)| +2 \max |\partial_{x_{i,n}} \psi(x)|^2,
\end{equation}
which from~\lem{err1}  is upper bounded by
\begin{equation}
\left(\frac{2}{\sqrt{5}}+\frac{2}{3} \right)k_{\max}^2 \left(\frac{k_{\max}}{\pi} \right)^{\eta D}= \left(\frac{6+2\sqrt{5}}{3\sqrt{5}}\right)k_{\max}^2 \left(\frac{k_{\max}}{\pi} \right)^{\eta D}.\label{eq:2derivbd}
\end{equation}
Now substituting~\eq{2derivbd} into~\eq{intbd72} yields
\begin{equation}
\left|\int_S |\psi(y(x))|^2 \mathrm{d}x^{\eta D}-1\right| \le  h^2\left(\frac{(6+2\sqrt{5})\eta D}{72\sqrt{5}}\right)k_{\max}^2 \left(\frac{k_{\max}L}{\pi} \right)^{\eta D}.\label{eq:discinterror}
\end{equation}
\eq{discinterror} is then at most $\tilde \delta$ if
\begin{equation}
h \le \sqrt{\frac{72\sqrt{5}\tilde \delta}{(6+2\sqrt{5})\eta D}}\frac{1}{k_{\max}} \left(\frac{k_{\max}L}{\pi} \right)^{-\eta D/2}.\label{eq:hbd0}
\end{equation}
Thus under this assumption on $h$ we have that
\begin{equation}
\sqrt{\int_S |\psi(y(x))|^2\mathrm{d}x^{\eta D}} \left|1-\frac{1}{\sqrt{\int_S |\psi(y(x))|^2 \mathrm{d}x^{\eta D}}} \right|\le \sqrt{1+\tilde \delta} \left( \frac{1}{\sqrt{1-\tilde \delta}} - 1 \right) \label{eq:69}
\end{equation}
If we assume $\tilde \delta \le 1/2$ then it is easy to verify that
\begin{equation}
\sqrt{1+\tilde \delta} \left( \frac{1}{\sqrt{1-\tilde \delta}} - 1 \right) \le \sqrt{\frac{3}{2}}\tilde{\delta}.
\end{equation}
Thus if we wish the upper bound in the error given in~\eq{69} to be at most $\delta$ it suffices to take $\tilde \delta = \sqrt{\frac{2}{3}}\delta$ and similarly $\delta \le \sqrt{\frac{3}{8}}$ implies our assumption on $\tilde \delta$.  The result then follows from substituting this choice of $\tilde \delta$ into~\eq{hbd0}, minimizing and using the fact that $(72 \sqrt{10/3})/(6+2\sqrt{5}) \approx 12.6 > 9$. 
\end{proof}

Combining  \cor{total} and \lem{normalization} allows us to prove \thm{totalab}. 
\begin{proofof}{\thm{totalab}}
We use the triangle inequality to break the simulation error into two terms corresponding to the results of \cor{total} and \lem{normalization}, respectively.
\begin{align}
&\left| \int_S \phi^*(x) e^{-i H t} \psi(x) \,\mathrm{d}^{\eta D}x - \int_S \phi^*(x) e^{-i \tilde{H} t} \psi(y(x)) \,\mathrm{d}^{\eta D}x \bigg/ \sqrt{ \int_S  |\psi(y(x))|^2 \,\mathrm{d} x^{\eta D}} \right|\nn
&\qquad \le \left|\int_S \phi^*(x) e^{-i H t} \psi(x) \,\mathrm{d}^{\eta D}x -\int_S \phi^*(x) e^{-i \tilde{H} t} \psi(y(x)) \,\mathrm{d}^{\eta D}x\right| \nn
&\qquad ~+ \left| \int_S \phi^*(x) e^{-i \tilde{H} t} \psi(y(x)) \,\mathrm{d}^{\eta D}x  - \int_S \phi^*(x) e^{-i \tilde{H} t} \psi(y(x)) \,\mathrm{d}^{\eta D}x \bigg/ \sqrt{ \int_S  |\psi(y(x))|^2 \,\mathrm{d} x^{\eta D}} \right| \nn
&\qquad \le  \left[\frac{k_{\max} \eta D h}{2\sqrt{3}} + t \left( \frac{ \pi^{3/2} e^{2a[ 1- \ln 2]}}{\sqrt{4a+3}} \frac{\eta D k_{\max}^{2a+1} h^{2a-1}}{18 m } + \frac{h {\eta D} }{2} V^\prime_{\max} \right) \right] \left(\frac{k_{\max} L }{\pi} \right)^{\eta D/2} + \delta.\label{eq:thm4bd}
\end{align}
In order to be able to use~\lem{normalization} we must choose 
\be
h \le 3\sqrt{\frac{\min(\delta,\sqrt{3/8})}{\eta D}}\frac{1}{k_{\max}}\left(\frac{k_{\max} L}{\pi} \right)^{-\eta D/2}.\label{eq:hlb1}
\ee

Next we want to find a value of $h$ such that
\be 
\left[\frac{k_{\max} \eta D h}{2\sqrt{3}} +\frac{h {\eta D} t  }{2} V^\prime_{\max} \right] \left(\frac{k_{\max} L }{\pi} \right)^{\eta D/2}<\left[\frac{k_{\max} \eta D h}{2} +\frac{h {\eta D} t  }{2} V^\prime_{\max} \right] \left(\frac{k_{\max} L }{\pi} \right)^{\eta D/2}\le\delta\label{eq:thm4bd2}
\ee
Thus we additionally require that 
\be 
h \le \frac{2\delta}{\eta D\left(k_{\max} +V'_{\max}t \right)}\left(\frac{k_{\max} L }{\pi} \right)^{-\eta D/2}.\label{eq:hselect}
\ee
We would like to make a uniform choice of $h$ in the theorem and to this end it is clear that $2\delta \le 3\sqrt{\min(\delta,\sqrt{3/8})}$ for $\delta \le 1/2$.  Thus since $\eta D\ge 1$ and $V'_{\max}t \ge 0$ it follows that~\eq{hselect} implies~\eq{hlb1} under our assumptions.  We therefore take~\eq{hselect} as $h$.

We then want to bound
\be
 \frac{ \pi^{3/2} e^{2a[ 1- \ln 2]}}{\sqrt{4a+3}} \frac{\eta D k_{\max}^{2a+1} h^{2a-1}}{18 m } t \left(\frac{k_{\max} L }{\pi} \right)^{\eta D/2} < \frac{ \pi^{3/2} e^{-2a/3}}{\sqrt{7}} \frac{\eta D k_{\max}^{2a+1} h^{2a-1}}{18 m } t \left(\frac{k_{\max} L }{\pi} \right)^{\eta D/2} \le  \delta,\label{eq:thm4bd3}
\ee
which holds if $k_{\max} h< e^{1/3}$ and
\be
a \ge \frac{3}{2} \frac{\log\left(\frac{1}{18 \sqrt7}\frac{\pi^{3/2} \eta D t k_{\max} }{\delta m h}\right)+\eta D\log\left(\frac{k_{\max} L}{\pi} \right) / 2 }{1-3\ln(k_{\max} h)}.
\ee
Therefore, assuming the worst-case scenario for $a$ where $k_{\max} \in O(1/h)$ we have from this choice of $a$ and the value of $h$ chosen in~\eq{hselect} that there exists $a$ such that the overall error is at most $\delta$ and
\be
a\in O\left(\eta D \log(k_{\max} L) + \log\left(\frac{\eta^2 D^2 t k_{\max} (k_{\max} +V'_{\max}t)}{m\delta^2} \right) \right).
\ee
The requirement that $k_{\max} h < e^{1/3}$ is then implied by~\eq{hselect}, $\delta \le 1/2$ and
\begin{equation}
k_{\max} L > \pi (2e^{-1/3})^{2/\eta D}.
\end{equation}

Then given these choices we have from~\eq{thm4bd},~\eq{thm4bd2} and~\eq{thm4bd3} that
\be
\left[\frac{k_{\max} \eta D h}{2\sqrt{3}} + t \left( \frac{ \pi^{3/2} e^{2a[ 1- \ln 2]}}{\sqrt{4a+3}} \frac{\eta D k_{\max}^{2a+1} h^{2a-1}}{18 m } + \frac{h {\eta D} }{2} V^\prime_{\max} \right) \right] \left(\frac{k_{\max} L }{\pi} \right)^{\eta D/2} + \delta \le 3\delta.
\ee
Hence the claim of the theorem holds for $\delta = \epsilon/3$.
\end{proofof}
The requirement on $a$ in \thm{totalab} is surprising: despite the fact that the derivatives of the wave function can scale exponentially with the number of particles $\eta$, as $k_{\max}^{\eta D}$, it is always possible to suppress this error with $a$ linear in $\eta$ and $D$, and in fact logarithmic in $k_{\max}$ and the inverse precision $1/\epsilon$.

However, the above work suggests that it is possible to get exponentially small upper bounds on the size of $h$ needed for the simulation if we make worst-case assumptions about the system and only impose a momentum cutoff.  It may seem reasonable to expect that such results come only from the fact that we have used worst-case assumptions and triangle inequalities to propagate the error.  However, in some cases this analysis is tight, as we show below.

Consider the minimum-uncertainty state for $D=1$,
\begin{equation}
\psi(x) =G(x):= \frac{\exp(-x^2/4\Delta x^2)}{\sqrt{\sqrt{2\pi} \Delta x}}\label{eq:G}.
\end{equation}
A simple exercise in calculus and the fact that $\Delta x \Delta p = \frac12$ shows that
\begin{equation}
\max_x |\partial_x \psi(x)| = \left(\frac{8}{\pi e^2} \right)^{1/4} \Delta p^{3/2} .
\end{equation}
This result shows that if we take $\Delta p \propto k_{\max}$ then it would follow that $|\partial_x \psi(x)|\in \Omega(k_{\max}^{3/2})$ which coincides with the upper bound in~\eq{diffpsieq}.  However, this is not directly comparable because the Gaussian function used here does not have compact support in either position or momentum.

We can deal with this issue of a lack of compact support in a formal sense by considering a truncated (unnormalized) minimum-uncertainty state:
\begin{equation}
\Psi(k) = \frac{\exp(-k^2/4\Delta p^2)}{\sqrt{\sqrt{2\pi} \Delta p}} {\rm Rect}\left(\frac{k}{2k_{\max}} \right),
\end{equation}
where ${\rm Rect}(x)$ is the rectangle function, ${\rm Rect}(x)=1$ if $x\in (-1/2,1/2)$, ${\rm Rect}(x)=0$ if $x\in \mathbb{R} \setminus [-1/2,1/2]$ and ${\rm Rect}(x)=1/2$ if $|x|=1/2$.  This function clearly has compact support in momentum space and thus satisfies the assumptions above.  We can rewrite this as
\begin{equation}
\Psi(k) = \psi(k) + \frac{\exp(-k^2/4\Delta p^2)}{\sqrt{\sqrt{2\pi} \Delta p}} \left({\rm Rect}\left(\frac{k}{2k_{\max}} \right)-1\right),
\end{equation}
By applying the Fourier transform and using standard bounds on the tail of a Gaussian distribution we then see that
\begin{equation}
|\partial_x \Psi(x)| = |\partial_x \psi(x)| + e^{-O(k_{\max}^2/\Delta p^2)}.\label{eq:Psijust}
\end{equation}
Thus we can take $k_{\max} \in \Theta(\Delta p)$ and make the approximation error that arises from truncating the support in momentum space exponentially small.  Thus these states have derivative $\Omega(k_{\max}^{3/2})$.

Now let us go beyond $\eta=1$ to $\eta>1$.  Since $\Delta x \propto 1/k_{\max}$ for this minimum-uncertainty state it then follows that $\partial_{x_{i}}\left(\Psi(x)^{\otimes \eta}\right) \in \Omega(k_{\max}(k_{\max})^{\eta/2})$ from~\eq{G} and~\eq{Psijust}. 
This means that the estimates of the derivatives used in the above results cannot be tightened without making assumptions about the quantum states in the system.  This further means that the exponential bounds cited above cannot be dramatically improved without either imposing energy cutoffs in addition to momentum cutoffs, or making appropriate restrictions on the initial state.

It may seem surprising that such a simple state should be so difficult to simulate. The reason for this is that we discretize into a uniform grid without making any assumptions about the state beyond a momentum cutoff: in this regard, uniform discretization is the basis choice corresponding to near-minimal assumptions about the system. Uniformly discretizating means that multi-dimensional Gaussian states becomes difficult to distinguish from a $\delta$ function as they becomes narrower and narrower, where with more knowledge of the system, we might be able to better parametrize the state, or to construct a better basis in which to represent the state, and thereby more efficiently simulate the system. Even when, as in this work, discretization is the first step in approximating evolution, Gaussian-like states can be efficiently simulated without exponentially small grid spacing for some Hamiltonians \cite{somma2015quantum}. More generally, there is the difficulty of not knowing which states might evolve into a high-derivative state at some future time, which is why we must also require the momentum cutoff to hold throughout the evolution. 

\cor{optimistic}, which we prove below, relies on the stricter assumption that the derivatives of the wave function obey $| \psi^{(r)}(x) | \le \beta k_{\max}^r / (\sqrt{2r+1} L^{\eta D/2})$ for the full duration of the simulation, rather than the worst-case bound $|\psi^{(r)}(x)| \le \frac{k_{\max}^r}{\sqrt{2r+1}}\left(\frac{k_{\max}}{\pi}\right)^{N/2}$ from \lem{err1} that was used in \thm{totalab}.

\begin{proofof}{\cor{optimistic}}
The proof follows from the exact same steps used to prove~\thm{totalab}.  By taking $|\partial^r_x \psi(x)|\in O( k_{\max}^r / (\sqrt{2r+1} L^{\eta D/2}))$ we can replicate all of the prior steps but substituting each $(k_{\max}/\pi)^{\eta D/2}$ with $\beta / L^{\eta D/2}$.  Thus each factor of $(k_{\max} L/\pi)^{\eta D/2}$ becomes $\beta$ after making this assumption.  This causes the additional additive term of $\eta D \log(k_{\max} L)$ to become zero in $a$ as well.  The claimed results then follow after making these substitutions.  For added clarity, we recapitulate the key steps in this argument below.

If we repeat the steps required in the proof of ~\cor{total} and~\lem{normalization} we see that
\begin{align}
&\left| \int_S \phi^*(x) e^{-i H t} \psi(x) \,\mathrm{d}^{\eta D}x - \int_S \phi^*(x) e^{-i \tilde{H} t} \psi(y(x)) \,\mathrm{d}^{\eta D}x \bigg/ \sqrt{ \int_S  |\psi(y(x))|^2 \,\mathrm{d} x^{\eta D}} \right|\nn
&\qquad \le  \beta \left[\frac{k_{\max} \eta D h}{2\sqrt{3}} + t \left( \frac{ \pi^{3/2} e^{2a[ 1- \ln 2]}}{\sqrt{4a+3}} \frac{\eta D k_{\max}^{2a+1} h^{2a-1}}{18 m } + \frac{h {\eta D} }{2} V^\prime_{\max} \right) \right] + \delta,\label{eq:cor5bd1}
\end{align}
if
\begin{equation}
h \le 3\sqrt{\frac{\min(\delta,\sqrt{3/8})}{\beta^2\eta D}}\frac{1}{k_{\max}}.\label{eq:hlb12}
\end{equation}
Following the exact same reasoning as in the proof of~\thm{totalab},
\begin{equation}
\left[\frac{k_{\max} \eta D h}{2\sqrt{3}} +\frac{h {\eta D} t  }{2} V^\prime_{\max} \right]\le\frac{\delta}{\beta} \label{eq:cor5bd2},
\end{equation}
if
\be 
h \le \frac{2\delta}{\beta \eta D\left(k_{\max} +V'_{\max}t \right)}.\label{eq:hselect2}
\ee

Finally again following the same reasoning that if $k_{\max} h \le e^{1/3}$ then
\be
 \frac{ \pi^{3/2} e^{2a[ 1- \ln 2]}}{\sqrt{4a+3}} \frac{\eta D k_{\max}^{2a+1} h^{2a-1}}{18 m } t  \le  \frac{\delta}{\beta},\label{eq:cor5bd32}
\ee
for a value of $a$ that scales at most as
\be
a\in O\left( \log\left(\frac{\eta^2 D^2\beta^2 t k_{\max} (k_{\max} +V'_{\max}t)}{m\delta^2} \right) \right).\label{eq:abd2}
\ee
Thus~\eq{cor5bd1} is bounded above by at most $3\delta$ given these choices and we can take $\delta=\epsilon/3$ to make all the results hold.  The result then follows by noting that the most restrictive scaling for $h$ out of the three requirements we place on it is
\begin{equation}
h\in O\left(\frac{\delta}{\beta \eta D\left(k_{\max} +V'_{\max}t \right)} \right),
\end{equation}
and using the fact that $\delta \in \Theta(\epsilon)$ and the assumption that $\beta \in \Theta(1)$ both here and in~\eq{abd2}.
\end{proofof}

\section{Discussion}
Conventional lore in quantum chemistry simulation has long postulated that continuous-variable simulations of chemicals affords far better scaling with the number of electrons than second-quantized methods, at the price of requiring more qubits.  Given the recent improvements in simulation algorithms for both first- and second-quantized Hamiltonians it is important to address the efficiency of quantum simulations using similar optimizations for continuous-variable simulations.  We investigate this question and find that through the use of high--order derivative formulas, it is possible under some circumstances to perform simulations using a number of calls to unitary adders and the pairwise interaction oracle that scale as $\tilde{O}(\eta^2 t \log(1/\epsilon))$.  
This is better than the best rigorous bounds proven for basis-based first- and second-quantized schemes, which scale as $\tilde{O}(\eta^5 t\log(1/\epsilon))$~\cite{babbush2016exponentially,babbush2015exponentially} assuming the number of spin-orbitals is proportional to the number of particles.

When we consider the discretization error after only assuming a momentum cutoff in the problem, we quickly see that in worst-case scenarios it is possible for such simulations to require a number of operations that scales exponentially in $\eta D$.  We further show that the derivative scaling that leads to this worst-case behavior can appear for minimum-uncertainty states.  This shows that although continuous-variable simulations offer great promise for quantum simulation, there are other caveats that must be met before they can be said to be efficient.  This problem also exists, to some extent, in second-quantized methods where such problems are implicitly dealt with by assuming that a sufficiently large basis is chosen to represent the problem.

We also show that these issues do not arise for more typical states, that is, states that have support that is much broader than a minimum-uncertainty state.  This demonstrates that the problems that can emerge when a highly localized state is provided as input do not necessarily appear for typical states that would be reasonable for ground state approximation and further agrees with the results of decades of experience in classical simulation of position space Hamiltonians.

There are a number of interesting questions that emerge from this work.  The first point is that many of the challenges that these methods face arise because of the use of a bad basis to represent the problem.  It is entirely possible that these issues can typically be addressed on a case-by-case basis, by choosing clever representations for the Hamiltonian as is typical in modern computational chemistry. Investigating the role that more intelligent choices of basis have for such simulations is an important future direction for research.

One further issue that this work does not address is the complexity of initial state preparation.  This problem is addressed in part in other work on quantum simulation in real space \cite{kassal2008polynomial}, and some common many-body states such as Slater determinants are known to be preparable with cost polynomial in $\eta$ and $1 / \epsilon$ \cite{ward2009preparation}. However, the costs of preparing more general appropriately symmetrized initial states can be considerable for fermionic simulations.
More work is needed to address such issues since the relative ease of state preparation for second-quantized methods can also be a major selling point for such fermionic simulations. 

Another issue that needs to be addressed is that despite the fact that continuous quantum simulations of chemistry using a cubic mesh are much more logical qubit-intensive than second-quantized simulations, they need not require more physical qubits because the lion's share of physical qubits are taken up by magic state distillation in simulations~\cite{jones2012faster,reiher2016elucidating}.  Further work is needed to differentiate the resource requirements of these methods at a fault-tolerant level.

Looking forward, despite the challenges posed by adversarially-chosen initial states, our work reveals that under many circumstances highly efficient simulations are possible for quantum chemistry that have better scaling than existing approaches.  This approach further does not require approximations such as the Born-Oppenheimer approximation to function, and thus can be straightforwardly applied in situations where such approximations are inappropriate.  Along these lines, it is important to develop a diverse arsenal of methods to bring to bear against simulation problems and understand the strengths as well as the limitations of each method.  It is our firm belief that as new approaches such as ours develop, quantum simulation will be thought of less as an algorithm and more as its own field of science that is viewed on the same level as numerical analysis, computational physics or quantum chemistry.

Finally, we note that a new linear combination-based technique \cite{low2016hamiltonian} allows the multiplicative factors in the cost to be separated if the grid spacing $h$ is fixed. This reduces the number of queries to the potential energy oracle to $\tilde O(\eta^2 t + \log(1/\epsilon)$. In general, however, the grid spacing may depend on $\epsilon$, removing this improvement.

\begin{acknowledgements}
We would like to acknowledge the Telluride Science Research Center for hosting us during the early phases of this project. I.\ D.\ K.\ thanks Peter J.\ Love and Guang Hao Low for stimulating discussions.  A.\ A.-G.\ acknowledges the Army Research Office under award W911NF-15-1-0256 and the Department of Defense Vannevar Bush Faculty Fellowship managed by the Office of Naval Research under award N00014-16-1-2008.
\end{acknowledgements}

\bibliographystyle{apsrev4-1}
%\bibliography{science}


\end{document}